\documentclass[12pt,a4paper,reqno]{amsart}
\usepackage[T1]{fontenc}

\usepackage{amsmath,amssymb,a4}

\usepackage{esint} 

\usepackage{graphicx}\unitlength 1mm

\usepackage{listings}
\usepackage{xcolor}

\definecolor{mapleinput}{rgb}{0.5,0.0,0.0}
\definecolor{maplemath}{rgb}{0.0,0.0,1.0}
\definecolor{maplewarning}{cmyk}{0.0,1.0,.0.0,0.0}
\lstnewenvironment{MapleInput}{%
\lstset{
	basicstyle=\bfseries\upshape\ttfamily\color{mapleinput},
	numberstyle=\bfseries\upshape\ttfamily\color{mapleinput},
	breaklines=true,
	columns=flexible,
	breakautoindent=false,
	breakindent=0pt,
	xleftmargin=4ex,
	numbers=left,
	stepnumber=1,
	numbersep=1.3ex,
	aboveskip=\smallskipamount,
	belowskip=\smallskipamount,
}%
}{%
}
\lstnewenvironment{MapleOutput}{%
\lstset{
	basicstyle=\upshape\ttfamily\color{maplemath},
	breaklines=true,
	columns=flexible,
	breakautoindent=false,
	breakindent=0pt,
	xleftmargin=2ex,
	numbers=none,
	aboveskip=0pt,
	belowskip=\smallskipamount,
}%
}{%
}
\newenvironment{MapleMath}{%
\color{maplemath}\upshape\rmfamily%
\setlength{\abovedisplayskip}{0ex}%
\setlength{\abovedisplayshortskip}{\abovedisplayskip}%
\setlength{\belowdisplayskip}{\medskipamount}%
\setlength{\belowdisplayshortskip}{0ex}%
\csname gather*\endcsname}{\csname endgather*\endcsname%
{\hrule height 0pt}%
\ignorespacesafterend}
\newcommand{\Hlog}[2]{\operatorname{Hlog}\left( #1, \left[ #2 \right] \right)}

\usepackage{url}
\usepackage[numbers,sort&compress]{natbib}
\newcommand{\eprintlink}[2]{\href{#1}{\color{blue}#2}}
\definecolor{links}{rgb}{0,0.3,0}
\usepackage[colorlinks=true,citecolor=blue,urlcolor=links,breaklinks=true]{hyperref}

\newcommand{\td}[1][]{\mathrm{d}^{#1}} 
\newcommand{\R}{\mathbb{R}} 
\newcommand{\C}{\mathbb{C}} 
\newcommand{\N}{\mathbb{N}} 
\newcommand{\Z}{\mathbb{Z}} 
\newcommand{\abs}[1]{\left\lvert #1 \right\rvert}

\newcommand{\iu}{\mathrm{{i}}} 
\DeclareMathOperator{\tr}{tr} 
\newcommand{\defas}{\mathrel{\mathop:}=} 
\newcommand{\ImPart}{\mathop{\mathrm{Im}}}
\newcommand{\RePart}{\mathop{\mathrm{Re}}}

\newcommand{\intbar}{\fint}
\newcommand{\textintbar}{\intbar}
\newcommand{\HT}{\mathcal{H}} 
\DeclareMathOperator{\Li}{Li} 

\DeclareMathOperator*{\arctanD}{arctan}

\newcommand{\cN}{\mathcal{N}}
\newcommand{\cZ}{\mathcal{Z}} 
\newcommand{\bigO}{\mathcal{O}}

\newcommand{\HyperInt}{\href{http://bitbucket.org/PanzerErik/hyperint/}{\texttt{\textup{HyperInt}}}}

\newtheorem{Proposition}{Proposition}
\newtheorem{Theorem}[Proposition]{Theorem}
\newtheorem{Lemma}[Proposition]{Lemma}
\newtheorem{Conjecture}[Proposition]{Conjecture}
\newtheorem{Remark}[Proposition]{Remark}
\newtheorem{Corollary}[Proposition]{Corollary}

\newcommand{\CO}{\Lambda^2}

\begin{document}

\title{Lambert-W solves the noncommutative $\Phi^4$-model}

\author{Erik Panzer}

\address{All Souls College, University of Oxford, OX1 4AL Oxford, 
United Kingdom} 

\email{erik.panzer@all-souls.ox.ac.uk}

\author{Raimar Wulkenhaar}

\address{Mathematisches Institut der WWU, 
Einsteinstr.\ 62, 48149 M\"unster, Germany}

\email{raimar@math.uni-muenster.de}

\begin{abstract}
The closed Dyson-Schwinger equation for the 2-point 
function of the noncommutative $\lambda \phi^4_2$-model is
rearranged into the boundary value problem for a sectionally
holomorphic function in two variables. We prove an exact formula 
for a solution in terms of Lambert's $W$-function. This solution is 
holomorphic in $\lambda$ inside a domain which contains 
$(-1/\log 4,\infty)$.
Our methods include the Hilbert transform, perturbation series and Lagrange-B\"{u}rmann resummation.
\end{abstract}


\maketitle

\section{Introduction}

In this paper we solve an integral equation that determines the
2-point function of a two-dimensional, noncommutative quantum field
theory. Our solution involves the Lambert function \cite{Lambert,
  Knuth} defined by $W(z) e^{W(z)}=z$ and is given in
\begin{Theorem}
	\label{thm:Lambert}%
The non-linear integral equation
for functions $G_\lambda\colon \R_+\! \times \R_+ \to \R$, 
\begin{align}
(1+a+b) G_\lambda(a,b)
= 1
&+ \lambda \int_0^{\infty} \!\!\! \td p
\;\Big(\frac{G_\lambda(p,b)-G_\lambda(a,b)}{p-a}
+ \frac{G_\lambda(a,b)}{1+p}\Big)
\nonumber%
\\
&+ \lambda \int_0^{\infty} \!\!\! \td q \;\Big(
\frac{G_\lambda(a,q)-G_\lambda(a,b)}{q-b}+ \frac{G_\lambda(a,b)}{1+q}\Big)
\label{NLIQ}\\
&- \lambda^2
\int_0^{\infty} \!\!\! \td p\int_0^{\infty} \!\!\! \td q\;
	\frac{G_\lambda(a,b)G_\lambda(p,q)-G_\lambda(a,q)G_\lambda(p,b)}{(p-a)(q-b)},
\nonumber%
\end{align}
admits for any real $\lambda>0$ a solution of the form 
\begin{align}
G_\lambda(a,b)&=G_\lambda(b,a)=\frac{
	(1+a+b)\exp(N_\lambda(a,b))
}{
	\big(a+\lambda W_0\big(\frac{1}{\lambda} e^{(1+b)/\lambda}\big)\big)
	\big(b+\lambda W_0\big(\frac{1}{\lambda} e^{(1+a)/\lambda}\big)\big)
}\;,\quad\text{where}
\label{Gab-finalnew}%
\\
N_\lambda(a,b)& \defas \frac{1}{2\pi\iu}\int_{-\infty}^\infty \td t\;
\log\Big(1-\frac{\lambda\log(\frac{1}{2}{-}\iu t)}{a+\frac{1}{2}+\iu t}\Big)
\frac{\partial}{\partial t} 
\log\Big(1-\frac{\lambda\log(\frac{1}{2}{+} \iu t)}{b+\frac{1}{2}-\iu t}\Big)
\label{eq:Nlambda}%
\end{align}
and $W_0$ denotes the principal branch of the Lambert function. 
This solution and equation \eqref{NLIQ} extend analytically
to a domain of complex $\lambda$
which contains the real interval $(-\frac{1}{\log 4},\infty)$.
In particular, $G_{\lambda}(a,b)$ from \eqref{Gab-finalnew} is the unique solution of \eqref{NLIQ} that is analytic at $\lambda=0$, and its radius of convergence in $\lambda$ is $1/\log 4 \approx 0.72$.
\end{Theorem}
As by-product we establish identities involving the Lambert function:
\begin{subequations}
\begin{gather}
\int_0^{\lambda} \frac{\td t}{t} 
	\frac{1}{1+W_0(\tfrac{1}{t} e^{1/t+a/\lambda})	}
= \log a - \log \left(
	\lambda W_0\big( \tfrac{1}{\lambda} e^{(1+a)/\lambda} \big)
	-1\right),
\label{eq:L-LambertInt}%
\\
\int_1^\infty \! \frac{\td u}{\pi}\;
\frac{\mbox{\small${\displaystyle\arctan\displaylimits_{[0,\pi]}  
\Big(\frac{\lambda\pi}{a{+}u -\lambda \log (u{-}1)}\Big)}$}}{u+z}
= \log \left( \frac{
	z+\lambda \log(1+z)-a}{
	1{+}z-\lambda W_0(\frac{1}{\lambda} e^{(1+a)/\lambda})
}\right),
\label{eq:L-Lambert-int}%
\\
\int_1^\infty  \frac{\td u}{\pi} \bigg( \!\!
\arctan\displaylimits_{[0,\pi]}
\Big(\frac{\lambda\pi}{a{+}u {-}\lambda \log (u{-}1)}\Big)
-\frac{\lambda\pi}{u}\bigg)
= \lambda W_0\big(\tfrac{1}{\lambda} e^{(1+a)/\lambda} \big)-1-a.
\label{eq:K-Lambert-int}%
\end{gather}
\end{subequations}
These are 
valid for all $a,\lambda \in \R_+$ and
$z\in \C\setminus {(-\infty,-1]}$. We also prove variants
of \eqref{eq:L-Lambert-int} and \eqref{eq:K-Lambert-int} for
$-1<\lambda<0$ which are not a simple continuation. 

We explain in section~\ref{sec:DSE} how the integral
equation~\eqref{NLIQ} arises from a quantum field theory model on a
noncommutative geometry. In sec.~\ref{sec:gakhov} we rewrite
\eqref{NLIQ} as a boundary value problem for a sectionally holomorphic
function in two variables which can partially be integrated to a
function $\tau(a)$ of a single variable. In
sec.~\ref{sec:perturbative} we determine the first terms of a formal
power series for $\tau(a)$ in $\lambda$. These terms are surprisingly simple and
allow us to guess the full formal power series in $\lambda$.  We resum
this series in sec.~\ref{sec:resummation}, prove that our guess is
correct and convert the result into the manifestly symmetric form
\eqref{Gab-finalnew}--\eqref{eq:Nlambda}. In
sec.~\ref{sec:continuation} we identify the holomorphic extension of
our solution and thereby finish the proof 
of Theorem~\ref{thm:Lambert}.

The integral \eqref{eq:Nlambda} is suitable for numeric evaluation of the function $N_\lambda(a,b)$. In sec.~\ref{sec:finalintegral} we discuss its perturbative expansion into Nielsen polylogarithms \cite{Nielsen:DilogarithmusVerallgemeinerungen}.
Some concluding remarks are collected in sec.~\ref{sec:discussion}.

\section{Dyson-Schwinger equation for the 2-point function}
\label{sec:DSE}

The $\lambda\phi^{\star 4}$-model with harmonic propagation 
on the $2$-dimensional Moyal plane is defined by the action functional \cite{GW:phi42matrixBase}
 \begin{equation}
	S(\phi)
	=\frac{1}{8\pi}
	\int_{\R^2} \td[2] x \left(
		\frac{1}{2}\,\phi\star \big(-\Delta +4\Omega^2 \abs{\Theta^{-1}x}^2 +\mu^2 \big)\phi
		+ \frac{\lambda}{4}\,\phi\star \phi\star \phi \star \phi
	\right),
	\label{eq:action}%
 \end{equation}
 where the mass $\mu$, coupling constant $\lambda$ and oscillator frequency $\Omega$ are real numbers.
The star denotes the Moyal product with deformation matrix 
 $\Theta=\left(\begin{smallmatrix}
	 0       & \theta  \\ 
	 -\theta & 0       \\
   \end{smallmatrix}\right)$
for a parameter $\theta>0$, defined by the oscillatory integral
\begin{equation*}
	(\phi\star \psi)(x)
	=
	\int_{\R^2\times \R^2} \frac{\td[2] y\, \td[2] k}{(2\pi)^2} 
	\;\phi(x+\tfrac{1}{2}\Theta k) \psi(x+y) 
	\,e^{\iu \langle y,k\rangle}\;.
\end{equation*}
There exists a family of ``matrix basis functions'' $f_{mn}$ on $\R^2$, indexed by $m,n \in \N_0$, which satisfy 
$(f_{mn} \star f_{kl})(x) = \delta_{nk} f_{ml}(x)$,
$\overline{f_{mn}(x)} =f_{nm}(x)$ and 
$\int_{\R^2} \td[2] x\;f_{mn}(x)=2\pi \theta \delta_{mn}$. 
For details see \cite{GraciaBondia:1987kw} or \cite[Appendix~A]{GW:phi42matrixBase}. The resulting correspondence
\begin{equation}
	\phi(x)
	=\sum_{m,n=0}^\infty \Phi_{mn} f_{mn}(x)
	\qquad \Leftrightarrow \qquad 
	\Phi_{kl}
	= \frac{1}{2\pi\theta} \int_{\R^2} \td[2] x \; (\phi \star f_{lk})(x)
	\label{eq:Frechet}%
\end{equation}
 defines an isomorphism of Fr\'echet algebras between Schwartz 
 functions with Moyal product and infinite matrices with rapidly 
 decaying entries. 
 This isomorphism extends to Moyal products between other classes of functions. 
 Real functions $\phi$ are represented by self-adjoint matrices $\Phi$. 

At critical frequency $\Omega=1$, the matrix basis functions satisfy 
\begin{equation*}
	(-\Delta + 4 \abs{\Theta^{-1}x}^2 ) f_{mn}
	= \tfrac{4}{\theta}(m+n+1) f_{mn}\;.
\end{equation*}
Therefore, at $\Omega=1$ the isomorphism \eqref{eq:Frechet} leads to
\begin{equation}
	S(\phi)
	\equiv S(\Phi)
	\stackrel{\Omega=1}{=} V\cdot \tr \left(E \Phi^2+\frac{\lambda}{4}\Phi^4 \right)\;,
	\label{action-matrix}%
\end{equation}
where $V=\frac{\theta}{4}$ and $E_{mn}=E_n \delta_{mn}$ is the diagonal matrix with
$E_n=\frac{\mu^2}{2}+\frac{n+1/2}{V}$. For the next steps matrix 
 sizes are restricted to $m,n\leq \cN$. Now the (Fourier
 transform of the) partition function is well-defined:
\begin{equation*}
	\cZ(J) 
	\defas \int D\Phi\;\exp \Big(-S(\Phi)+\iu V \tr (J\Phi) \Big)\;,
\end{equation*}
 where $D\Phi$ is the Lebesgue measure on
 $\R^{(\cN+1)^2}$ and $J$ another matrix with rapidly decaying entries for
 $\cN\to \infty$. As usual for matrix
 models, the logarithm of $\cZ(J)$ has an expansion into
 boundary cycles $(p^\beta_{N_\beta+1}\equiv p^\beta_{1}$),
\begin{equation*}
	\log \frac{\cZ(J)}{\cZ(0)} 
	=\sum_{B=1}^\infty
	\sum_{N_B\geq\ldots\geq N_1\geq 1} \frac{V^{2-B}}{S_{N_1\dots N_B}} 
\!\sum_{p^1_1,\dots,p^B_{N_B}} \!\!\!\!\!
	G_{|p^1_1\dots p^1_{N_1}|\dots|p^B_1\dots p^B_{N_B}|}
	\prod_{\beta=1}^B \Big(
		\prod_{j_\beta=1}^{N_\beta} \iu
		J_{p^\beta_{j_\beta}p^\beta_{j_\beta+1}}
	\Big).
\end{equation*}

The following Dyson-Schwinger equations for the 2- and 4-point functions 
were derived in \cite[equations~(3.4) and (3.7)]{GW:Phi44nonnon}, here 
with $\bigO(\tfrac{1}{V})$-terms 
suppressed:\footnote{A behaviour $\sum_n \sim V$ is assumed so that the sums 
in \eqref{Gab-old} are kept.}
\begin{align}
G_{|ab|}
&= \frac{1}{E_a+E_b}+\frac{(-\lambda)}{E_a+E_b} \Big(
	\frac{1}{V} \sum_{m=0}^{\cN} G_{|ab|}G_{|am|}
	-\frac{1}{V}\sum_{a\neq n=0}^{\cN} \frac{G_{|ab|}-G_{|nb|}}{E_a-E_n}
\Big)
\;,
\label{Gab-old}%
\\
G_{|abcd|}&= (-\lambda) \frac{
	G_{|ab|}G_{|cd|}-G_{|ad|}G_{|cb|}
}{
	(E_a-E_c)(E_b-E_d)
}\;.
\label{Gabcd}%
\end{align}
These equations rely on a Ward identity discovered in
\cite{Disertori:2006nq}. By the same techniques one can derive another
Dyson-Schwinger equation for the 2-point function (again with
$\bigO(\frac{1}{V})$-terms suppressed):
\begin{equation}
G_{|ab|}
= \frac{1}{E_a+E_b}
+ \frac{(-\lambda)}{E_a+E_b} \Big(
	\frac{1}{V^2} \sum_{m,n=0}^{\cN} G_{|bamn|}
	+\frac{1}{V} \sum_{n=0}^{\cN} G_{|ab|}(G_{|an|}+G_{|nb|})
\Big)\;.
\label{Gab-SD}%
\end{equation}
This Dyson-Schwinger equation has an obvious graphical interpretation. 
The proof combines \cite[eqs.~(3.2)~and~(3.3)]{GW:Phi44nonnon} in our conventions and for $a\neq b$ to 
\begin{equation*}
	G_{|ab|}
	= \frac{1}{E_a+E_b}
	+ \frac{(-\lambda)}{(E_a+E_b)V^3} 
	\sum_{m,n=0}^{\cN} \left.
		\frac{
			\partial^4(\exp(\log \frac{\cZ(J)}{\cZ(0)}))
		}{
			\partial J_{ba}\partial J_{am}
			\partial J_{mn}\partial J_{nb}
		}
	\right|_{J=0}\;.
\end{equation*}
Generically the $J$-differentiations yield the 4-point function 
$G_{|bamn|}$ to be summed over $m,n$. But there are also the cases 
$m=b$ or $n=a$ where a disconnected product of 2-point functions
contributes in $\exp(\log(\cZ(J)))$, producing the last terms in 
\eqref{Gab-SD}. Other contributions such as $m=n=a$ and $m=n=b$ are
$\bigO(\frac{1}{V})$-suppressed. 

We eliminate $\sum_n G_{|an|},\sum_n G_{|nb|}$ in \eqref{Gab-SD} via 
\eqref{Gab-old} and express $G_{|bamn|}$ in \eqref{Gab-SD} via 
\eqref{Gabcd}. The sums can safely exclude $m=b$ and $n=a$
because these contribute with an exceeding $\frac{1}{V}$-factor which
is anyway ignored. We have thus proved:

\begin{Lemma}
The 2-point function of the $\lambda\phi^{\star4}$-model with harmonic
propagation on 2-dimensional Moyal space satisfies in matrix
representation (with cut-off $\cN$, up 
to $\bigO(\frac{1}{V})$-corrections) 
the following closed Dyson-Schwinger equation: 
\begin{multline*}
G_{|ab|}= \frac{1}{E_a+E_b}\bigg\{
1 
-\frac{(-\lambda)^2}{V^2} 
\sum_{a\neq n=0}^{\cN} \sum_{b\neq m=0}^{\cN} 
\frac{
	G_{|ab|}G_{|nm|}-G_{|am|}G_{|nb|}
}{
	(E_a-E_n)(E_b-E_m)
}
\\
-\frac{(-\lambda)}{V}\sum_{a\neq n=0}^{\cN} 
\frac{G_{|ab|}-G_{|nb|}}{E_a-E_n}
- \frac{(-\lambda)}{V}\sum_{b\neq m=0}^{\cN} 
\frac{G_{|ab|}-G_{|am|}}{E_b-E_m}
\bigg\}\;.
\end{multline*}
\end{Lemma}
\noindent
Compared with \eqref{Gab-old} this equation is manifestly symmetric in 
$a,b$ and contains the $G$-quadratic terms in a more regular way.

As in \cite{GW:Phi44nonnon} we take a combined limit $\cN,V\to \infty$ with
$\frac{\cN}{V}= \CO$ fixed.
In this limit, $E_n\mapsto \frac{\mu^2}{2}+p$ with 
$p=\frac{n}{V}=\frac{n}{\cN}\frac{\cN}{V}\in [0,\CO]$. The 2-point function 
becomes a function $G(a,b)$ of real arguments $a,b\in [0,\CO]$, 
and the densitised sums converge to principal value Riemann 
integrals\footnote{%
	For H\"older-continuous $G(a,b)$, this is in fact just the ordinary integral. But since we will pull the numerators apart later, we write principal values already here.} over $[0,\CO]$:
\begin{multline}
(a+b+\mu^2) G(a,b)
= 1
+ \lambda \intbar_0^{\CO} \!\!\! \td p\;\frac{G(p,b)-G(a,b)}{p-a}
+ \lambda \intbar_0^{\CO} \!\!\! \td q\; \frac{G(a,q)-G(a,b)}{q-b}
\\
- \lambda^2
\intbar_0^{\CO} \!\!\! \td p
\intbar_0^{\CO} \!\!\! \td q\;
\frac{G(a,b)G(p,q)-G(a,q)G(p,b)}{(p-a)(q-b)}\;.
\label{Gab-integral}%
\end{multline}
This equation is exact: the previously ignored
$\bigO(\frac{1}{V})$-terms are strictly absent.

\section{A boundary value problem \`a la Gakhov}
\label{sec:gakhov}

We employ a method from Gakhov's book
\cite{Gakhov} on \emph{boundary value problems}:\footnote{%
	RW would like to thank Alexander Hock for pointing out this reference.
}
\begin{Proposition}\label{prop:Gakhov}%
The closed integral equation~\eqref{Gab-integral} for the 2-point 
function of the $\lambda\phi^{\star4}$-model with harmonic
propagation on 2-dimensional Moyal space is in the scaling limit 
$\cN\to \infty$, keeping $\frac{\cN}{V}=\CO$
fixed, equivalent to the following boundary value problem: 
Define a holomorphic function on $(\C\setminus [0,\CO])^2$ by
\begin{multline}
\Psi(z,w) \defas
	\mu^2+z+w
	+\lambda \log \frac{\CO-z}{(-z)}
	+\lambda \log \frac{\CO-w}{(-w)}
	\\
	\quad + 
	\lambda^2 \int_{0}^{\CO} \!\!\! \td p\int_{0}^{\CO} \!\!\! \td q \;
	\frac{G(p,q)}{(p-z)(q-w)}.
	\label{eq:Psi}%
      \end{multline}
  Let $\Psi^{\alpha\beta}(a,b) \defas
\lim_{\epsilon\to 0} 
\Psi(a+\iu\alpha\epsilon,b+\iu\beta\epsilon)$
denote its boundary values for the signs $\alpha,\beta = \pm 1$ of the infinitesimal imaginary parts. 
Then for any $a,b \in (0,\CO)$,
\begin{equation}
\Psi^{++}(a,b)\Psi^{--}(a,b) = \Psi^{+-}(a,b)\Psi^{-+}(a,b)\;.
	\label{Q}%
\end{equation}
\end{Proposition}
\begin{proof}
The boundary values on $w,z\rightarrow (0,\CO)$ of the 
integral in \eqref{eq:Psi},
\begin{equation*}
	Q(z,w)
	\defas \frac{1}{\pi^2}\int_{0}^{\CO} \!\!\! \td p
	\int_{0}^{\CO} \!\!\! \td q \; \frac{G(p,q)}{(p-z)(q-w)}\;,
\end{equation*}
are given by the Sokhotski-Plemelj formulae
\begin{multline}
Q^{\alpha \beta}(a,b)
=\frac{1}{\pi^2}\intbar_{0}^{\CO} \!\!\! \frac{\td p}{p-a}
\intbar_{0}^{\CO} \!\!\! \frac{\td q}{q-b} \;G(p,q)
\\
+
\frac{\iu \alpha}{\pi}\intbar_{0}^{\CO} \!\!\! \td q \;
\frac{G(a,q)}{q-b}
+
\frac{\iu \beta}{\pi}\intbar_{0}^{\CO} \!\!\! \td p \;
\frac{G(p,b)}{p-a}
-\alpha\beta G(a,b).
\label{eq:Plemelj}%
\end{multline}
Inserting these boundary values and those of the logarithm
\begin{equation*}
	\lim_{\epsilon\rightarrow 0}\;
\lambda \log \left. \frac{\CO-z}{(-z)} \right|_{z=a+\iu\alpha\epsilon}
= \lambda \log \frac{\CO-a}{a}+\alpha\iu\pi \lambda
\end{equation*}
into \eqref{Q} gives a formula which is easily rearranged into 
\eqref{Gab-integral}.
\end{proof}

We find it remarkable that an interacting quantum field theory,
which when expanded into Feynman graphs evaluates to 
Nielsen polylogarithms (see later),
admits such a simple presentation.
Unfortunately we are not aware of a solution theory for
such boundary value problems, and therefore develop an ad hoc approach
in the sequel.

First, assuming that $G(a,b)$ falls off asymptotically 
like in the four-dimensional case \cite{GrosseWulkenhaar:OnTheFixed}, 
we achieve a well-defined limit $\lim_{\CO\to \infty} \Psi$ upon 
choosing the bare mass 
\begin{equation*}
\mu^2=1-2\lambda \log(1+\CO)
=1-\lambda\int_0^{\CO} \frac{\td p}{1+p}
-\lambda\int_0^{\CO} \frac{\td q}{1+q}.
\end{equation*}
This choice was made in \eqref{NLIQ}. 
Next, we require $\lambda$ and consequently $G(p,q)$ to be real, which
implies $\overline{\Psi(z,w)} = \Psi(\bar{z},\bar{w})$. Now
\eqref{Q} can be written as $|\Psi^{++}|=|\Psi^{+-}|$. Hence there is
a continuous real 
function $\tau_a(b)$, not expected to be symmetric under 
$a\leftrightarrow b$, with
\begin{equation}
	\Psi^{++}(a,b) e^{-\iu \tau_a(b)}=\Psi^{+-}(a,b) e^{\iu \tau_a(b)}\;.
	\label{eq:Gakhov-real}%
\end{equation}
We denote the \emph{finite Hilbert transform} of a function $f$ in 
one or two variables as
\begin{align*}
	\HT_{a}^\Lambda[f(\bullet)] &\defas
\frac{1}{\pi} \;
\intbar_{0}^{\CO} \!\! \frac{\td p}{p-a}
f(p)
	\quad\text{and} \\
\HT_{a,b}^\Lambda[f(\bullet,\bullet)] &\defas
\frac{1}{\pi^2} \;
\intbar_{0}^{\CO} \!\! \frac{\td p}{p-a}
\intbar_{0}^{\CO} \!\! \frac{\td q}{q-b} \; f(p,q).
\end{align*}
The substitutions \eqref{eq:Plemelj} and \eqref{eq:Psi} then turn \eqref{eq:Gakhov-real} into
\begin{align*}
& 
\big(1+a+b+\lambda\log\tfrac{\CO-a}{(1+\CO)a}
+\lambda\log\tfrac{\CO-b}{(1+\CO)b}
+2\pi\lambda\iu
\\
&\qquad\quad +\lambda^2\pi^2 \big(\HT^\Lambda_{a,b}[G(\bullet,\bullet)]
-G(a,b)+\iu \HT^\Lambda_b[G(a,\bullet)]
+\iu \HT^\Lambda_a[G(\bullet,b)]\big)
\big)e^{-\iu \tau_a(b)}
\\
&= \big(1+a+b+\lambda\log\tfrac{\CO-a}{(1+\CO)a}
+\lambda\log\tfrac{\CO-b}{(1+\CO)b}
\\
&\qquad\quad
+ \lambda^2\pi^2 \big(\HT^\Lambda_{a,b}[G(\bullet,\bullet)]
+G(a,b)+\iu \HT^\Lambda_b[G(a,\bullet)]
-\iu \HT^\Lambda_a[G(\bullet,b)]\big)\big)
e^{\iu \tau_a(b)}.
\end{align*}
Comparing the real- and imaginary parts identifies this relation with the system
\begin{subequations}\begin{align}
G(a,b) \cot (\tau_a(b))- \HT_b^\Lambda[G(a,\bullet)]  
	&= \frac{1}{\lambda\pi} \;,
	\label{Carl1}%
\\
\mathcal{G}(a,b) \cot (\tau_a(b)) - \HT_b^\Lambda [\mathcal{G}(a,\bullet)] 
	&= \frac{
		1{+}a{+}b+\lambda\log\frac{\CO-a}{a(1+\CO)}
-\lambda\log(1{+}\CO)
	}{
		\lambda^2 \pi^2
	},\label{Carl2}%
\end{align}\end{subequations}
where $\mathcal{G}(a,b) \defas
\frac{1}{\lambda\pi}+\HT_a^\Lambda[G(\bullet,b)]$. These equations are
Carleman-type singular integral equations for which a solution theory
is developed e.g.\ in \cite[\S4.4]{Tricomi}. The method is sensitive
to the sign of $\lambda$. We restrict ourselves to 
$\lambda>0$; only in the very end we will pass via analytic continuation to 
some complex and negative $\lambda$.

The starting point of the solution theory for Carleman equations is
\begin{Proposition}
	For any H\"{o}lder continuous function
$\tau\colon [0,\CO]\longrightarrow \R$, one has
\begin{subequations}
\begin{align}
 \HT_b^{\Lambda}\big[e^{\HT^\Lambda_\bullet[\tau]} \sin \tau(\bullet)\big] 
&= e^{\HT^\Lambda_b[\tau]} \cos \tau(b) -1,
& & \text{if $0<b<\CO$, and}
\label{Tricomi-18}%
\\
\int_0^{\CO} \!\!\frac{\td p}{\pi} \;
\frac{e^{\HT^\Lambda_p[\tau]} \sin \tau(p)}{p-b}
&= \exp\Big(\int_0^{\CO} \!\! \frac{\td p}{\pi}
\;\frac{\tau(p)}{p-b} 
\Big) -1
& & \text{for $b<0$ or $b>\CO$.}
\label{Tricomi-outside}%
\end{align}
We also have the corollary
\begin{equation}
\int_0^{\CO} \!\!\td p\;
e^{\pm \HT^\Lambda_p[\tau]} \sin \tau(p) 
= \int_0^{\CO} \!\! \td p 
\;\tau(p)\;.
\label{eq:tauidentity}%
\end{equation}
\end{subequations}
\end{Proposition}
\begin{proof}
Consider the function 
$
	\varphi(z)
	\defas 
	\exp\big( 
		\frac{1}{\pi}
		\int_0^{\CO} 
		\td p \;\frac{\tau(p)}{p-z}
	\big)
	-1
$.
It is holomorphic on $\C\setminus [0,\CO]$ and it decays at least like $\abs{z}^{-1}$ at infinity.
The fundamental property
\begin{equation*}
	\frac{1}{\pi} \intbar_{-\infty}^\infty \frac{\td p}{p-b}  \; 
	\ImPart (\varphi(p+\iu \epsilon))
	= \RePart (\varphi(b+\iu \epsilon))
	\qquad\text{(where $b\in\R$)}
\end{equation*}
of the Hilbert transform over $\R$ implies \eqref{Tricomi-18} and \eqref{Tricomi-outside}, because
$\ImPart (\varphi(p+\iu \epsilon))
=e^{\HT^\Lambda_p[\tau]} \sin \tau(p)$
inside the support $[0,\CO]$, and  
$\ImPart (\varphi(p+\iu \epsilon))=0$
for $p \in \R\setminus [0,\CO]$. 
Multiply \eqref{Tricomi-outside} for $\pm\tau$ with $b>\CO$ and take the limit $b\to \infty$ to obtain \eqref{eq:tauidentity}.
\end{proof}
With \eqref{Tricomi-18} we can immediately solve \eqref{Carl1} in terms of the angle function to 
\begin{equation}
G(a,b)=e^{\HT_b^\Lambda[\tau_a(\bullet)]} \frac{\sin \tau_a(b)}{\lambda\pi}\;.
\label{Gab-real}%
\end{equation}
In order to describe a meaningful two-point function $G(a,b)>0$ it
remains to verify that $\tau_a(b)\in [0,\pi]$ is consistent with
continuity of $\tau$. To create the term linear in $b$ on the 
right-hand side of \eqref{Carl2}, observe that
\begin{equation*}
\HT_b^\Lambda [\bullet \cdot G(a,\bullet)]
=
\frac{1}{\pi} \intbar_0^{\CO}\!\!\! \left( \frac{b\; \td p}{p-b} +\td p \!\right) G(a,p)
= b \HT_b^\Lambda [ G(a,\bullet) ] 
+ \frac{1}{\lambda\pi^2} \int_0^{\CO} \!\!\td p\;\tau_a(p)
\end{equation*}
by virtue of \eqref{Gab-real} and \eqref{eq:tauidentity}.
It is then easily checked with \eqref{Carl1} that
\begin{equation*}
\mathcal{G}(a,b) \defas \frac{
	1+a+b+\lambda\log\frac{\CO-a}{a(1+\CO)}
	+\frac{1}{\pi} \int_0^{\CO} \td p \;\tau_a(p)
	-\lambda \log(1{+}\CO)	
}{
	\lambda \pi
} G(a,b)
\end{equation*}
solves \eqref{Carl2}. However, this is not necessarily the only
solution because (see \cite{Tricomi})
\begin{equation}
\frac{e^{\HT^\Lambda_b[\tau_a]}\cos
  \tau_a(b)}{\CO-b}-
\HT^\Lambda_b\Big[
	\frac{e^{\HT^\Lambda_\bullet[\tau_a]}\sin
	\tau_a(\bullet)}{\CO-\bullet}
  \Big]=0\;.
\label{eq:Carl-hom}
 \end{equation}
Therefore, also a shift $\mathcal{G}(a,b) \mapsto \mathcal{G}(a,b) 
+\frac{\CO h(a)}{\lambda\pi(\CO-b)} G(a,b)$ solves 
\eqref{Carl2} for any function $h(a)$. We will not exploit this
freedom and consider the solution arising from setting $h(a)=  0$. 
We will prove that it gives rise to an analytic solution $G_\lambda(a,b)$ of 
\eqref{NLIQ} in a neighbourhood of $\lambda=0$, but there might be 
other, non-analytic solutions.

Recall that we had 
$\mathcal{G}(a,b) = \frac{1}{\lambda\pi}+\HT_a^\Lambda[G(\bullet,b)]$, 
which is equal to $G(a,b) \cot(\tau_b(a))$ by \eqref{Carl1} for symmetric $G(a,b)=G(b,a)$. We thus need
\begin{equation}
\label{cottauba-co}%
\cot \tau_b(a)
= \frac{1}{\lambda\pi}\bigg\{
	1+a+b+\lambda\log\frac{\CO-a}{a(1+\CO)}
	+ \frac{1}{\pi}\int_0^{\CO} \!\!\! \td p\; 
\Big(\tau_a(p)	-\frac{\lambda\pi}{1+p}\Big)\bigg\}\;. 
\end{equation}
For every angle function subject to this constraint, we have just constructed a solution \eqref{Gab-real} of the non-linear integral equation \eqref{Gab-integral}.
Analogous identities for the $\lambda\phi^{\star 4}_4$-model were derived in \cite{GW:Phi44nonnon}, by a completely different strategy. Note that \eqref{cottauba-co} implies in particular that $\cot \tau_b(a)= \frac{b}{\lambda\pi} +\cot \tau_0(a)$.

Taking the limit $\CO\to \infty$, we obtain a solution of our initial problem:
\begin{Proposition}
\label{Prop:Ilambda}%
The integral equation \eqref{NLIQ} is solved for $\lambda>0$ by
\begin{equation*}
	G_\lambda(a,b) \defas 
	\frac{\sin \tau_b(a)}{\lambda\pi}
	e^{\HT_a[\tau_b(\bullet)]}, \quad\text{where}\quad
	\HT_a[f(\bullet)]
	\defas
	\frac{1}{\pi}\! \textintbar_{0}^{\infty} \!\! \frac{\td p}{p-a} f(p)
\end{equation*}
is the one-sided Hilbert transform,
provided that there is a  continuous solution of
\begin{equation}\label{cottauba}\begin{split}
\tau_b(a)
	&\phantom{:}=\arctan\displaylimits_{[0,\pi]} 
\Big(\frac{\lambda\pi}{1+a+b-\lambda\log a + I_\lambda(a)}\Big)\;,
	\quad\text{where}
\\
	I_\lambda(a)
	& \defas \frac{1}{\pi}\int_0^\infty \!\! \td p\; \Big(
		\tau_a(p)-\frac{\lambda\pi}{1+p}	
	\Big)\;.
\end{split}\end{equation}
\end{Proposition}

\section{Perturbative solution}
\label{sec:perturbative}

We try to solve \eqref{cottauba} as a formal power series in
$\lambda$. This strategy leads surprisingly far. The solution 
clearly starts with 
$\tau_b(a)=\frac{\lambda\pi}{1+a+b}+\bigO(\lambda^2)$
which gives the $2$-point function $G_\lambda(a,b)=\frac{1}{1+a+b}+\bigO(\lambda)$.
It follows then that
\begin{align*}
I_\lambda(a) 
	&= \lambda \int_0^{\infty} \!\! \td p\left(
		\frac{1}{1+p+a}-\frac{1}{1+p}
	\right) + \bigO(\lambda^2)
	 =-\lambda \log (1+a)+\bigO(\lambda^2)\;,
\\
\tau_b(a)
	&=\frac{\lambda\pi}{1+a+b}
	+\frac{\lambda^2\pi}{(1+a+b)^2}\big(\log a +\log(1+a)\big)
	+\bigO(\lambda^3)\;.
\end{align*}
By elementary techniques, we continued these integrations and obtained $I_\lambda(a)$ up to corrections of order $\bigO(\lambda^5)$. The result is strikingly simple and structured:
\begin{align}
I_\lambda(a) &= ({-}\lambda) \log(1{+}a)
+\tfrac{({-}\lambda)^2}{1{+}a} ((1{+}a){+}a) \tfrac{\log(1{+}a)}{a}
\nonumber
\\
&+\tfrac{({-}\lambda)^3}{(1{+}a)^2} 
\big(((1{+}a){+}a)\tfrac{\log(1{+}a)}{a}
- ((1{+}a)^2{+}a^2) \tfrac{(\log(1{+}a))^2}{2a^2}
\big)
\nonumber
\\
& + 
\tfrac{({-}\lambda)^4}{(1{+}a)^3}\big( 
((1{+}a){+}a) \tfrac{\log(1{+}a)}{a} 
- \big(2((1{+}a)^2{+}a^2) +a((1{+}a){+}a)\big)\tfrac{(\log(1{+}a))^2}{2a^2}
\nonumber
\\
&\qquad\qquad + ((1{+}a)^3{+}a^3) \tfrac{(\log(1{+}a))^3}{3a^3}\big)
+\bigO(\lambda^5)\;.
\label{eq:Ilambda4}
\end{align}
Interestingly, the Hilbert transform of this simple angle function
generates higher transcendental functions, for example polylogarithms
$\Li_n(z) = \sum_{k>0} z^k/k^n$, which show up in the $2$-point
function. To second order, we find
\begin{equation}\begin{split}
\label{Gab-pert}%
	G_\lambda(a,b)
	&=\frac{1}{1+a+b}
	+\frac{\lambda}{(1+a+b)^2} \Big(\log(1+a)+\log(1+b) \Big)
\\ &
	-\frac{\lambda^2}{(1+a+b)^2}\Big(
		\frac{1+2a}{a(1+a)}\log(1+a)
		+\frac{1+2b}{b(1+b)}\log(1+b)
	\Big)
\\ &
	+\frac{\lambda^2}{(1+a+b)^3}\Big(
		\log^2(1+a)
		+\log(1+a)\log(1+b)
		+\log^2(1+b)
\\&\qquad\qquad\qquad\quad
		+\zeta(2)-\Li_2(-a)-\Li_2(-b)
	\Big)
	+\bigO(\lambda^3).
\end{split}\end{equation}

\subsection*{Higher orders}
As illustrated above, the perturbative calculation leads to
expressions that are rational linear combinations of logarithms.
Furthermore, they have only very simple singularities, confined to the
hyperplanes
\begin{equation*}
	a=0, \quad \quad b=0, \quad a+1=0, \quad b+1=0 \quad\text{and}\quad 1+a+b=0.
\end{equation*}
The integration theory on such hyperplane complements\footnote{%
  Our case is isomorphic to the moduli space
  $\mathfrak{M}_{0,5}$ of genus zero curves with $5$ marked points. } is completely understood
\cite{Brown:PhD,Panzer:PhD} in terms of iterated integrals, and
computer implementations are available
\cite{Bogner:MPL,Panzer:HyperInt}.  We note that there is also an
alternative approach based on the toolbox of holonomic recurrences
\cite{AblingerBluemleinSchneider:GeneralizedHarmonicSumsAndPolylogarithms,Schneider:ModernSummation}.

Using {\HyperInt} \cite{Panzer:HyperInt}, it is straightforward to compute higher orders of $I_{\lambda}(a)$.\footnote{%
	{\HyperInt} can also find \eqref{Gab-pert} by integrating directly the perturbative expansion of \eqref{NLIQ}. Furthermore, {\HyperInt} computes $H=\int_0^{\infty} \frac{\td p}{p-a} \tau_b(p)$ as an integral over real $p$, adding an imaginary part $\iu\epsilon \delta_a$ to $a$. The Hilbert transform of $\tau_b(a)$ is thus the real part of $H$ (i.e.\ drop the $\delta_a$-term).
}
To give an illustration, note that the 2nd order of $I_\lambda(a)$ in \eqref{eq:Ilambda4} contributes, among several others, the term 
$\frac{2 \pi\lambda^3 \log (1+p) \log p}{(1 + a + p)^3}$
to the 3rd order of $\tau_a(p)$.
With
\begin{MapleInput}
read "HyperInt.mpl":
tau3 := 2*Pi*log(1+p)*log(p)/(1+p+a)^3;
I3 := hyperInt(tau3/Pi, p=0..infinity);
\end{MapleInput}
one computes its contribution to the $\lambda^3$-coefficient of
$I_\lambda(a)$. The command
\begin{MapleInput}
fibrationBasis(I3,[a]);
\end{MapleInput}
\begin{MapleMath}%
\frac{\Hlog{a}{0, -1}}{(1+a)^2} +\frac{\zeta_2}{(1+a)^2}
+\frac{\Hlog{a}{-1, -1}}{a^2}-\frac{2 \Hlog{a}{-1}}{a(1+a)}
\end{MapleMath}
expresses the result in terms of hyperlogarithms
$\Hlog{a}{\sigma,\tau}=\int_0^a \frac{\td z}{z-\sigma} \int_0^z \frac
{\td u}{u-\tau}$.
In this particular case, they are just polylogarithms:
$ \Hlog{a}{0,-1}= -\Li_2(-a)$, 
$ \Hlog{a}{-1,-1}= \frac{1}{2}(\log(1+a))^2$ and
$ \Hlog{a}{-1}= \log(1+a)=-\Li_1(-a)$.
Most strikingly, in the final result for $I_{\lambda}(a)$, the hyperlogarithms cancel almost completely and collapse to mere powers of the logarithm $\log(1+a)$.

In this way,\footnote{%
	At first we were unaware of \eqref{eq:tauidentity} and calculated the much harder $\int_0^\infty \td p \;e^{-\HT_p[\tau_a]}\sin \tau_a(p)$.
}
we computed all coefficients of $\lambda^{\leq 10}$ in $I_{\lambda}(a)$.
The results are of such striking simplicity and structure that we could
obtain an explicit formula:
\begin{Conjecture}
	With Stirling numbers $s_{n,k}$ of the first kind and sign $(-1)^{n-k}$,
\begin{equation}\begin{split}
\label{HyperInt}%
	I_\lambda(a)
	&= -\lambda \log(1+a) + 
	\sum_{n=1}^\infty  \lambda^{n+1} \left(
		\frac{(\log (1+a))^{n}}{na^{n}} 
		+\frac{(\log (1+a))^{n}}{n(1+a)^{n}}
	\right)
\\ & 
	+ \! \sum_{n=1}^\infty \frac{(n{-}1)! \lambda^{n+1}}{(1{+}a)^{n}} 
	\sum_{j=1}^{n-1} \sum_{k=0}^{n}
		({-}1)^{j}
		\frac{s_{j,n-k}}{k! j!} 
		\Big(\Big(\frac{1{+}a}{a}\Big)^{n-j} \!\! +1\Big)
		\big(\log(1{+}a)\big)^k
	.
\end{split}\end{equation}
\end{Conjecture}
In the next section, we will first find a closed expression for the sum \eqref{HyperInt} using the Lambert-$W$ function, and then prove the conjecture.

\section{Resummation and solution}
\label{sec:resummation}

\subsection{Resummation}

The Stirling numbers of first kind have generating function
\begin{align*}
	(1+z)^u
	=\sum_{n=0}^\infty \sum_{k=0}^n \frac{z^n}{n!} u^k s_{n,k}\;,
\quad 
	(-1)^j s_{j,n-k}
	=\frac{1}{(n{-}k)!}
	\left.
		\frac{\td[n-k]}{\td u^{n-k}} 
		\frac{\Gamma(j{-}u)}{\Gamma(-u)}
	\right|_{u=0}\;.
\end{align*}
Let \eqref{HyperInt}$_2$ be the 2nd line of
\eqref{HyperInt}. Writing also $(\log(1+a))^k=\frac{\td[k]}{\td u^k}
(1+a)^u \big|_{u=0}$, this line takes the form
\begin{equation}\begin{split}
\label{HyperInt-2}%
\eqref{HyperInt}_2&= 
\sum_{n=1}^\infty \frac{\lambda^{n+1}}{n(1+a)^{n}} 
\sum_{j=1}^{n-1} 
\Big(\Big(\frac{1+a}{a}\Big)^{n-j}
+1\Big)
\\
& \qquad \times 
\sum_{k=0}^{n} \binom{n}{k} \left(\frac{\td[n-k]}{\td u^{n-k}} 
\frac{\Gamma(j-u)}{j! \Gamma(-u)}\right)
\left. \left( \frac{\td[k]}{\td u^k} (1+a)^u \right)\right|_{u=0}\;.
\end{split}\end{equation}
The summation over $k$ gives for the 2nd line of \eqref{HyperInt-2}
\begin{align*}
	\frac{\td[n]}{\td u^{n}} 
\left.\left(\frac{\Gamma(j-u)}{j! \Gamma(-u)}(1+a)^u\right)
\right|_{u=0}
&= \frac{\td[n]}{\td u^{n}} 
\left(\frac{(-1)^j}{j!} (1+a)^j \frac{\td[j]}{\td a^j} (1+a)^u \right)
\bigg|_{u=0}
\\
&= \frac{(-1)^j}{j!} (1+a)^j \frac{\td[j]}{\td a^j} (\log(1+a))^{n}\;.
\end{align*}
This is inserted back into \eqref{HyperInt-2} and the 2nd line of \eqref{HyperInt}.
Now the first line of \eqref{HyperInt} is the missing case $j=0$ to extend $I_\lambda(a)$ to 
\begin{multline*}
I_\lambda(a)= -\lambda \log(1+a)
\\
+	\sum_{n=1}^\infty \frac{\lambda^{n+1}}{n} 
	\sum_{j=0}^{n-1} \left(
		\frac{1}{(1+a)^{n-j}}
		+\frac{1}{a^{n-j}}
	\right)
	\frac{(-1)^j}{j!} \frac{\td[j]}{\td a^j} (\log(1+a))^{n}\;.
\end{multline*}
Writing 
$
	\frac{1}{a^{n-j}}
	= \frac{(-1)^{n-1-j}}{(n-1-j)!}
	  \frac{\td[n-1-j]}{\td a^{n-1-j}}
	  \frac{1}{a}
$ 
and similarly for $\frac{1}{(1+a)^{n-j}}$, we thus arrive at
\begin{align}
	I_\lambda(a)
	&= -\lambda \log(1{+}a) 
	 + \sum_{n=1}^\infty \frac{({-}\lambda)^{n+1}}{n!} 
	   \frac{\td[n-1]}{\td a^{n-1}} \left(
	   	\frac{(\log(1+a))^{n}}{1+a}
		+\frac{(\log(1+a))^{n}}{a}
	\right)\!\!\!
\nonumber\\
	&= \sum_{n=1}^\infty \frac{\lambda^n}{n!} 
	   \frac{\td[n-1]}{\td a^{n-1}}
	   (-\log(1+a))^{n}
	 -\lambda \sum_{n=1}^\infty \frac{\lambda^n}{n!} 
	   \frac{\td[n-1]}{\td a^{n-1}}
	   \frac{(-\log(1+a))^{n}}{a}\;.
\label{eq:Ilambda-diff}%
\end{align}

There are several ways to sum these series. The most efficient
approach seems to be the Lagrange-B\"{u}rmann inversion formula
\cite{Lagrange, Buermann}:
\begin{Theorem}
	\label{thm:Lagrange-inversion}%
Let $\phi(w)$ be analytic at $w=0$ with $\phi(0)\neq 0$ and 
$f(w) \defas \frac{w}{\phi(w)}$. Then the inverse $g(z)$ of $f(w)$
with $z=f(g(z))$ is analytic at $z=0$ and given by
	\begin{equation}
		g(z) = \sum_{n=1}^{\infty} \frac{z^n}{n!} 
\left.\frac{\td[n-1]}{\td w^{n-1}}\right|_{w=0} \phi(w)^n
		.
		\label{eq:Lagrange}%
	\end{equation}
More generally, if $H(z)$ is an arbitrary analytic function with $H(0)=0$, then
	\begin{equation}
		H(g(z)) = \sum_{n=1}^{\infty} \frac{z^n}{n!} \left.\frac{\td[n-1]}{\td w^{n-1}}\right|_{w=0} \Big( H'(w) \phi(w)^n \Big)
		.
		\label{eq:Buermann}%
	\end{equation}
\end{Theorem}
By virtue of \eqref{eq:Lagrange}, we see upon setting 
$z=\lambda$ and $\phi(w)=-\log(1+a+w)$ that the first summand 
in \eqref{eq:Ilambda-diff},
\begin{equation}
	K(a,\lambda)
	\defas
	\sum_{n=1}^\infty \frac{\lambda^n}{n!} 
	   \frac{\td[n-1]}{\td a^{n-1}}
	   (-\log(1+a))^{n}
	,
	\label{eq:K-def}%
\end{equation}
is the inverse of the function $\lambda(w)=-\frac{w}{\log(1+a+w)}$:
\begin{equation}
	K(a,\lambda)=-\lambda \log(1+a+K(a,\lambda))\;.
	\label{eq:K-FEQ}%
\end{equation}
This functional equation is easily solved in terms of Lambert-W \cite{Knuth},
\begin{equation} 
	K(a,\lambda)
	= \lambda \,W\left(
		\frac{1}{\lambda}
		e^{\frac{1+a}{\lambda}}
	\right)-1-a\;.
	\label{eq:K-Lambert}%
\end{equation}
For any $\lambda>0$ and $a\geq 0$, the solution is given by the
standard real branch $W_0$. We discuss in sec.~\ref{sec:continuation}
the extension of 
\eqref{eq:K-FEQ} and \eqref{eq:K-Lambert} to (certain) 
$a,\lambda\in \C$.

Let us now turn to the second summand (up to a 
factor $-\lambda$) in \eqref{eq:Ilambda-diff},
\begin{equation}
	L(a,\lambda)
	\defas
	\sum_{n=1}^\infty \frac{\lambda^n}{n!}
		\frac{\td[n-1]}{\td a^{n-1}}
		\frac{(-\log(1+a))^n}{a}
	.
	\label{eq:L-def}%
\end{equation}
This can directly be recognised as \eqref{eq:Buermann} with $H(w)=\log(1+w/a)$, such that
\begin{equation}
	L(a,\lambda)
	= \log\left( 1+ \frac{K(a,\lambda)}{a} \right)
	= \log \frac{\lambda W\left( \frac{1}{\lambda}
            e^{(1+a)/\lambda} 
\right)-1}{a}
	.
	\label{eq:L-Lambert}%
\end{equation}

\begin{Remark}[Alternative solution of \eqref{eq:Ilambda-diff}]
Express the multiple derivatives in \eqref{eq:Ilambda-diff} via
Cauchy's formula and insert $\frac{1}{n}=\int_0^1
\frac{\td t}{t} t^n$. This gives rise to a geometric series in $n$ and results in a
single residue located at the solution of $1{+}a=1{+}z+\lambda t
\log(1{+}z)$. In fact, this extends to a proof of the Lagrange-B\"urmann formula.
\end{Remark}

\begin{Remark}[Alternative solution II of \eqref{eq:Ilambda-diff}]
Starting from the series expansion
\begin{equation*}
	\frac{\log(1+a)^{n}}{n!}
	=\sum_{m=n}^\infty s_{m,n} \frac{a^m}{m!}
	\quad\text{valid for $\abs{a}<1$,}
\end{equation*}
again with Stirling numbers $s_{m,n}$ of first kind, we 
can expand \eqref{eq:K-def} and \eqref{eq:L-def} as\footnote{%
	It is well-known that the expansion of Lambert-W at infinity 
is related to Stirling numbers, see \cite{Knuth}. However, we did not 
find the precise form we obtain here in the literature.
}
\begin{equation*}
	K(a,\lambda) 
	= \sum_{n=1}^\infty\sum_{m=1}^\infty s_{m+n-1,n}
	(-\lambda)^n \frac{a^m}{m!}
	,\quad
	L(a,\lambda)
	= \sum_{n=1}^\infty\sum_{m=0}^\infty \frac{s_{m+n,n}}{m+n} 
	(-\lambda)^n  \frac{a^m}{m!}.
\end{equation*}
The recursion relation $n s_{n,k}= s_{n,k-1}-s_{n+1,k}$ then proves the differential equations
\begin{equation}\begin{split}
\label{PDE}%
	\left[
		(1+a+\lambda) \frac{\partial}{\partial a} 
		+\lambda\frac{\partial}{\partial\lambda}
	\right] K(a,\lambda)
	&= K(a,\lambda)-\lambda
	\quad\text{and}
\\
	\left[ 
		a \frac{\partial}{\partial a} 
		+\lambda\frac{\partial}{\partial\lambda}
	\right] L(a,\lambda)
	&= \frac{\partial K(a,\lambda)}{\partial a}\;.
\end{split}\end{equation}
It is straightforward to check that the solutions \eqref{eq:K-Lambert}
and \eqref{eq:L-Lambert} solve this system. In fact, the solution
\eqref{eq:K-Lambert} of the first equation in \eqref{PDE} is found by Maple \cite{Maple}
and fixed through the boundary conditions $K(0,\lambda)=0$ and
$K(a,0)=0$.  

Changing variables $(a,\lambda) \mapsto
(v=\frac{a}{\lambda},\lambda)$, the 2nd equation in \eqref{PDE}
becomes an ordinary differential equation with respect to $\lambda$,
in which $v$ is merely a parameter:
\begin{equation*}
	\lambda\frac{\td}{\td\lambda} L(v,\lambda)
	=-1+ W'(\tfrac{1}{\lambda} e^{v+\frac{1}{\lambda}})
	 \tfrac{1}{\lambda} e^{v+\frac{1}{\lambda}}
	=-\frac{1}{1+W(\tfrac{1}{\lambda} e^{v+\frac{1}{\lambda}})}
	.
\end{equation*}
Again, it is easily checked that \eqref{eq:L-Lambert} solves this equation. However, quadrature from the boundary value $L(v,0)=0$ provides the solution in a different form,
\begin{equation}
	L(a,\lambda)
	= - \int_0^{\lambda} \frac{\td t}{t} 
	\frac{1}{
		1+W(\tfrac{1}{t} e^{1/t+a/\lambda})
	}\;.
\end{equation}
We conclude the non-trivial identity \eqref{eq:L-LambertInt} given in the
beginning.  
\end{Remark}

\subsection{Proof of the solution}

The proof that our guess \eqref{HyperInt} is correct relies on
\begin{Lemma}\label{Lemma:J}%
For any $a,\lambda\geq 0$ and 
$z \in \C\setminus {(-\infty,-1]}$ one has\footnote{These formulae 
are reproduced as \eqref{eq:K-Lambert-int} and
\eqref{eq:L-Lambert-int} in the beginning.} 
\begin{subequations}
\begin{gather}
\int_0^\infty \! \frac{\td u}{\pi}\; 
\frac{\mbox{\small{$\displaystyle 
\arctan\displaylimits_{[0,\pi]} 
\Big(\frac{\lambda\pi}{1+a+u -\lambda \log u}\Big)$}}}{
1+u+z}
= 
\log \left( \frac{z+\lambda \log(1{+}z)-a}{
1{+}z{-}\lambda
  W_{0}(\frac{1}{\lambda} e^{(1+a)/\lambda})}\right)\;,
\label{eq:J1}%
\\
\int_0^\infty \! \frac{\td u}{\pi}\; 
\Big(\!
\arctan\displaylimits_{[0,\pi]}  
\Big(\frac{\lambda\pi}{1{+}a{+}u {-}\lambda \log u}\Big)
-\frac{\lambda\pi}{1{+}u}\Big)
= \lambda W_0\Big(\frac{1}{\lambda} e^{(1+a)/\lambda}\Big) -1-a\;.
\label{eq:J2}%
\end{gather}
\end{subequations}
\end{Lemma}
\begin{proof}
Note that $u\mapsto \frac{1+a+u -\lambda \log u}{\lambda\pi}$ is a
convex continuous function so that its reciprocal cannot vanish 
in $(0,\infty)$. This makes $u\mapsto \arctan\displaylimits_{[0,\pi]} 
\big(\frac{\lambda\pi}{1+a+u -\lambda \log u}\big)$ a continuous
function which can be written as imaginary part of a complex logarithm.
Let $J_\lambda(z,a)$ be the integral in \eqref{eq:J1} and 
$\gamma_\epsilon^+$ be the curve in the complex plane which encircles 
the positive reals clockwise at distance $\epsilon$. 
We choose $\epsilon$ such that $\gamma_\epsilon^+$ 
separates $-1-z$ from 
$\R_+$ and rewrite the integral as
\begin{align*}
J_\lambda(z,a) &=\frac{1}{2\pi\iu }\int_{\gamma_\epsilon^+} \!\! \td w\; 
\frac{\log\big(1-\frac{\lambda \log (-w)}{1+a+w}\big)}{
w+1+z}+o(\epsilon)\;.
\end{align*}
The outer logarithm expands into the power series 
\begin{align*}
J_\lambda(z,a) &=
-\sum_{n=1}^\infty \frac{\lambda^n}{n} 
\frac{1}{2\pi\iu }\int_{\gamma_\epsilon^+} \!\! \td w\; 
\frac{(\log (-w))^n}{(1+a+w)^n (w+1+z)}
\end{align*}
in $\lambda$ with radius of convergence at
least $\frac{1-\epsilon}{\sqrt{\pi^2+|\log \epsilon|^2}}$.
We close $\gamma^+_\epsilon$ by a large circle, which does not
contribute to the integral, and temporarily assume $z\neq a$. 
The residue theorem picks up the pole
at $w=-a-1$ of order $n$ and the simple pole at $w=-z-1$:
\begin{align*}
J_\lambda(z,a) 
&=-\sum_{n=1}^\infty \frac{\lambda^n}{n!} \frac{\td[n-1]}{\td w^{n-1}}
\frac{(\log (-w))^n}{(w+1+z)}\Big|_{w=-1-a}
- \sum_{n=1}^\infty \frac{\lambda^n}{n} 
\frac{(\log (-w))^n}{(1+a+w)^n}\Big|_{w=-1-z}
\\
&=\sum_{n=1}^\infty \frac{\lambda^n}{n!} \frac{\td[n-1]}{\td w^{n-1}}
\frac{(-\log (1+a+w))^n}{(1+z-(1+a+w))}\Big|_{w=0}
+ \log\Big(1-\frac{\lambda \log(1+z)}{a-z}\Big)\;.
\end{align*}
The rhs is independent of $\epsilon$ so that the equality holds 
exactly, but depending on $a-z,1+z,1+a$ for a 
possibly smaller, but still non-zero, radius of convergence in $\lambda$. 
The sum is again a Lagrange-B\"urmann formula \eqref{eq:Buermann} 
for $H(w)= \log \frac{a-z}{a+w-z}$ and $w\mapsto 
\lambda W_0(\frac{1}{\lambda}e^{(1+a)/\lambda})-1-a$. 
The result rearranges into \eqref{eq:J1} which is thus proved 
for $a\neq z$ and sufficiently small $\lambda$. Since both sides are
real-analytic in $\lambda>0$ and holomorphic in $z$,
\eqref{eq:J1} extends to any $\lambda>0$ and 
$z \in \C\setminus {(-\infty,-1]}$.

Following the same steps as above, the integral 
in \eqref{eq:J2} can be written as
\begin{align*}
	&\lim_{\epsilon\to0}
	\frac{1}{2\pi\iu}\int_{\gamma^+_\epsilon} \!\!
\td w \;\Big\{ 
\log \Big(1-\frac{\lambda\log (-w)}{1+a+w}\Big)
+ \frac{\lambda \log (-w)}{1+w}
\Big\}
\\
&=-\sum_{n=1}^\infty \frac{\lambda^n}{n!} \frac{\td[n-1]}{\td w^{n-1}}
(\log (-w))^n\Big|_{w=-1-a}
=K(a,\lambda)
\end{align*}
from \eqref{eq:K-def}. We conclude with \eqref{eq:K-Lambert}.
\end{proof}
\begin{Proposition}
\label{prop:Ilambda}%
The pair of equations \eqref{cottauba} is for any $a,\lambda \geq 0$ 
solved by the resummation $I_\lambda(a)=K(a,\lambda)-\lambda L(a,\lambda)$ from \eqref{eq:K-Lambert} and \eqref{eq:L-Lambert}, that is
\begin{equation}
\tau_a(p)=\arctan\displaylimits_{[0,\pi]} \bigg(\frac{\lambda\pi}{a+ 
 \lambda \,W_0\big(\frac{1}{\lambda} e^{\frac{1+p}{\lambda}}\big)
-\lambda \log \big( \lambda W_0\big( \frac{1}{\lambda} 
        e^{\frac{1+p}{\lambda}} \big)-1\big)}\bigg)\;.
\label{eq:solution-tau}%
\end{equation}
\end{Proposition}
\begin{proof}
The assertion \eqref{eq:solution-tau} solves \eqref{cottauba} precisely if
\begin{equation}
\int_0^\infty \!\!\! \td p\: \Big(\frac{\tau_a(p)}{\pi}
-\frac{\lambda}{1{+}p}\Big)
=  \lambda \,W_0\Big(\frac{e^{\frac{1+a}{\lambda}}}{\lambda}\Big)
    -1-a
-\lambda \log \frac{\lambda W_0\big( \frac{
	e^{\frac{1+a}{\lambda}}}{\lambda} \big)-1}{a}
	.
\label{eq:Ilambda-final}%
\end{equation}
Introduce a cut-off at $\CO+\lambda\log(1+\CO)$ and substitute
$p=u+\lambda\log(1+u)$ with inverse 
$u=\lambda W(\frac{1}{\lambda}e^{(1+p)/\lambda})-1$ into the lhs of 
\eqref{eq:Ilambda-final}. 
It then becomes
\begin{equation*}
\int_0^{\CO} \!\!\! \td u \left[
\Big(1+\frac{\lambda}{1+u}\Big) 
\frac{\arctan\displaylimits_{[0,\pi]}\big( 
	\frac{\lambda\pi}{a+1+u-\lambda \log u}\big)}{\pi}
-\frac{\lambda}{1{+}u}
\right]
- \lambda\int_{\CO}^{\CO+\lambda\log(1+\CO)}\!\! \frac{\td u}{1{+}u}
\;.
\end{equation*}
The second term vanishes in the limit $\Lambda\to \infty$, and the
first term evaluates with \eqref{eq:J1} at $z=0$ and 
\eqref{eq:J2} to the rhs of 
\eqref{eq:Ilambda-final}.
\end{proof}

\subsection{Proof of Theorem 1}
\label{sec:simplify}

With the solution of \eqref{cottauba} established in
\eqref{eq:solution-tau} and with \eqref{Gab-real} we have also
achieved the proof that
\begin{align}
&G_\lambda(a,b) \label{Gab-final}
\\
&=\frac{\displaystyle \exp
\bigg[
\frac{1}{\pi} \intbar_0^\infty \!\!\! 
\frac{\td p}{p{-}b}\arctan\displaylimits_{[0,\pi]}
\bigg( 
\frac{\lambda\pi}{a+
\lambda W_0\big(\tfrac{e^{(1+p)/\lambda}}{\lambda}\big)
- \lambda \log\big(\lambda W_0\big(
\tfrac{e^{(1+p)/\lambda}}{\lambda}\big){-}1\big)}
\bigg)
\bigg]
}{\displaystyle \sqrt{(\lambda\pi)^2+ 
\big[a+
\lambda W_0\big(\tfrac{e^{(1+b)/\lambda}}{\lambda}\big)
- \lambda \log\big(\lambda W_0\big(
\tfrac{e^{(1+b)/\lambda}}{\lambda}\big){-}1\big)\big]^2}
}
\nonumber%
\end{align}
solves the original integral equation \eqref{NLIQ} for $\lambda>0$.
It remains to prove that the
Hilbert transform in \eqref{Gab-final} simplifies to the form given in 
Theorem~\ref{thm:Lambert}.
\begin{Lemma}
For all $a,b\geq 0$ and $\lambda>0$ the following identities hold:
\begin{subequations}
\begin{align}
\HT_b\!\left[\arctan\displaylimits_{[0,\pi]} \!\left(\!\frac{\lambda\pi}{
 1{+}a{+}{\bullet}-\lambda \log (\bullet)}\right)\right]
  	&= 
	\log \!\left(\!\frac{
		\sqrt{(1{+}a{+}b{-}\lambda \log b)^2 +(\lambda \pi)^2}
	}{
		b + \lambda W_0\big(\frac{1}{\lambda} e^{(1+a)/\lambda}\big)
	}\right),
\label{eq:HTArctanLog}%
\\
	\!\frac{1}{2\pi\iu}
	\int_{\gamma_{\epsilon}^+}  \frac{\td w}{w{-}b}
	\log\!\left(1-\frac{\lambda \log({-}w)}{1{+}a{+}w}\right)
  &= 
	\log \left(\frac{
		1+a+b
	}{
		b + \lambda W_0\big(\frac{1}{\lambda} e^{(1+a)/\lambda}\big)
	}\right).
\label{eq:logW0-path-integral}%
\end{align}
\end{subequations}
\end{Lemma}
\begin{proof}
The Hilbert transform \eqref{eq:HTArctanLog} is the real part of \eqref{eq:J1} at
$z=-1-b-\iu\epsilon$. 
The proof of \eqref{eq:logW0-path-integral} follows the
same strategy; the difference is that $w=b$ and $\R_+$ in the proof of
Lemma~\ref{Lemma:J} are
both on the same side of $\gamma_\epsilon^+$. Therefore, after series
expansion only the residue at $w=-1-a$ contributes:
\[
\frac{1}{2\pi\iu}
\int_{\gamma_{\epsilon}^+}  \frac{\td w}{w{-}b}
\log\left(1-\frac{\lambda \log({-}w)}{1{+}a{+}w}\right)
= -\sum_{n=1}^\infty \frac{\lambda^n}{n!} \frac{\td[n-1]}{\td w^{n-1}}
\frac{(\log (-w))^n}{w-b}\Big|_{w=-1-a}\;.
\]
This is a Lagrange-B\"urmann formula \eqref{eq:Buermann} for 
$H(w)=\log \frac{1+a+b}{1+a+b+w}$. 
\end{proof}

We can finally express $\HT_b[\tau_a(\bullet)]$ through the symmetric
integral given in \eqref{eq:Nlambda}:
\begin{Proposition}
For any $\lambda>0$ the Hilbert transform of \eqref{eq:solution-tau} 
evaluates to
\begin{align*}
\HT_b[\tau_a(\bullet)]
&=  \log \sqrt{(\lambda\pi)^2+\Big(a+\lambda W_0(\tfrac{1}{\lambda}e^{(1+b)/\lambda})
-\lambda \log\big(\lambda W_0(\tfrac{1}{\lambda}e^{(1+b)/\lambda})-1\big)\Big)^2}
\\
& + \log\left( \frac{(1+a+b) \exp(N_\lambda(a,b))}{
(b+\lambda W_0(\tfrac{1}{\lambda}e^{(1+a)/\lambda}))
(a+\lambda W_0(\tfrac{1}{\lambda}e^{(1+b)/\lambda}))}\right),
\end{align*}
where $N_\lambda(a,b)=N_\lambda(b,a)$ is given by
\begin{equation}
N_\lambda(a,b)
= \frac{1}{2\pi\iu} \int_{\gamma_\epsilon^+} \!\! 
\td w \log\Big(1-\frac{\lambda\log(-w)}{1+a+w}\Big)
\frac{\partial}{\partial w} 
\log\Big(1-\frac{\lambda\log(1{+}w)}{1{+}b-(1{+}w)}\Big)
\label{eq:Nlambda-1}%
\end{equation}
or equivalently by \eqref{eq:Nlambda}.
In particular, formula \eqref{Gab-finalnew} of 
Theorem~\ref{thm:Lambert} follow. 
\end{Proposition}
\begin{proof}
The Hilbert transform $\HT_b[\tau_a(\bullet)]$ 
in \eqref{Gab-final} 
is the real part of the standard integral when shifting
$b\mapsto b+\iu \epsilon$. We 
substitute $p=u+\lambda\log(1+u)$ with inverse 
$u=\lambda W_0(\frac{1}{\lambda}e^{(1+p)/\lambda})-1$ and 
rewrite this as 
a contour integral of a complex logarithm over $\gamma_{\epsilon/2}^+$
(to separate $\R_+$ from $b+\iu \epsilon$):
\begin{equation*}
\HT_b[\tau_a(\bullet)]
= \RePart\bigg( \frac{1}{2\pi\iu } \int_{\gamma_{\epsilon/2}^+} 
\!\!\! \td w\;
\frac{(1{+}w{+}\lambda)\;
\log \big(1-\frac{\lambda \log(- w)}{a{+}1{+}w} \big)}{(1{+}w)
(w{+}\lambda\log(1{+}w)-(b+\iu \epsilon))}\bigg)\;.
\end{equation*}
We subtract and add the
same integral over the curve  $\gamma_{2\epsilon}^+$. The difference
retracts to the residue at the solution of $
w+\lambda\log(1+w)=b$ which gives back a 
Lambert function. In the remaining integral we can shift $b+\iu
\epsilon\to b$ and obtain
\begin{equation*}\begin{split}
\HT_b[\tau_a(\bullet)]
&= \frac{1}{2\pi\iu} \int_{\gamma_\epsilon^+} \td w\;
\frac{\displaystyle (1{+}w{+}\lambda)\;
\log \Big(1-\frac{\lambda \log(- w)}{a{+}1{+}w} \Big)}{(1{+}w)
(w{+}\lambda\log(1{+}w)-b)}
\\
&+ \RePart \Big(\log 
 \Big(1-\frac{\lambda \log(- \lambda
   W_0(\frac{1}{\lambda}e^{(1+b)/\lambda})+1-\iu \epsilon')}{
a+ \lambda W_0(\frac{1}{\lambda}e^{(1+b)/\lambda})} \Big)\Big)\;.
\end{split}\end{equation*}
Adding $0=\text{rhs}-\text{lhs}$ of \eqref{eq:logW0-path-integral} identifies $N_\lambda(a,b)$ in 
\eqref{eq:Nlambda-1}. 

The integrand of \eqref{eq:Nlambda-1} is for all $a,b\geq
0$ holomorphic in $\lambda \in \C$ away from the two branch cuts $[0,\infty)$ and 
$(-\infty,-1]$. Therefore, $\gamma_\epsilon^+$ can be deformed into the 
line $-\frac{1}{2}+\iu t$. The result \eqref{eq:Nlambda}
exposes the symmetry $N_\lambda(a,b)=N_\lambda(b,a)$ through integration by parts. 
\end{proof}

Asymptotic formulae \cite{Knuth} of Lambert-W
give $\lambda W_0(\frac{1}{\lambda} e^{(1+a)/\lambda})\geq 
C_0 (1+a)^{1-\delta}$ for some $C_0 > 0$ and $0<\delta<\frac{1}{2}$. 
The function $N_\lambda(a,b)$ tends to $0$ for large $a,b$ which implies 
a bound $N_\lambda(a,b) \leq \log (C_1)$ for some $C_1>0$. 
We thus conclude from \eqref{Gab-finalnew} the bound 
$0<G_\lambda(a,b)\leq \frac{C_1}{(1+a+b)^{1-2\delta}}
\leq \frac{C_1}{(1+a)^{\frac{1}{2}-\delta}
(1+b)^{\frac{1}{2}-\delta}}$. This bound a posteriori justifies 
our assumption that all integrals converge for $\CO\to \infty$. 

\begin{Remark}
Another presentation uses \eqref{eq:J1} to write 
\begin{align*}
\frac{\td}{\td w} 
\log \Big(1 -  \frac{\lambda \log(1{+}w)}{b-w}\Big)
&= -
\frac{1}{\pi}\int_0^\infty \!\! \td u\; 
\frac{\arctan\displaylimits_{[0,\pi]} 
\big(\frac{\lambda\pi}{1+b+u -\lambda \log u}\big)}{(1+u+w)^2}
\\[-1ex]
& +
\frac{1}{w{-}(\lambda  W_0(\frac{1}{\lambda} e^{(1+b)/\lambda})-1)}
-\frac{1}{w-b}\;,
\end{align*}
for $\lambda>0$,
$b\geq 0$ and $w\in \C\setminus (-\infty,-1]$.
Inserted into \eqref{eq:Nlambda-1} one obtains with \eqref{eq:logW0-path-integral} a purely real and
manifestly symmetric integral representation 
\begin{equation}
G_\lambda(a,b)=\frac{\exp\Bigg(\displaystyle
{-}\!\!
\int_0^\infty \!\!\!\! \td u  \!\! \int_0^\infty \!\!\!\!\td v  \;\frac{
	{\arctanD\limits_{[0,\pi]}}
	\big(\frac{\lambda\pi}{1{+}a{+}v{-}\lambda \log v}\big)
	{\arctanD\limits_{[0,\pi]}}
	\big(\frac{\lambda\pi}{1{+}b{+}u{-}\lambda \log u}\big)
}{\pi^2(1+u+v)^2}\Bigg)}{
\lambda 
W_0(\frac{1}{\lambda}e^{(1+a)/\lambda}) +\lambda 
W_0(\frac{1}{\lambda}e^{(1+b)/\lambda})-1}\;.
\label{Gab-alternative}%
\end{equation}
\end{Remark}

\section{Holomorphic extension}
\label{sec:continuation}

We established a solution $G_{\lambda}(a,b)$ for the integral equation
\eqref{NLIQ} for real values of $a,b\geq 0$ and $\lambda>0$. We will
now discuss its analytic continuation in two regimes:
\begin{enumerate}
\item For real $a,b \geq 0$, we determine the maximal domain of
  analyticity for complex values of $\lambda$ (as alluded to in
  Theorem~\ref{thm:Lambert}).
\item Keeping $\lambda$ real, we can extend to complex values of $a$
  and $b$.
\end{enumerate}
First we will consider the function $K(a,\lambda)$ from
\eqref{eq:K-Lambert} that appears in $I_{\lambda}(a)$ and the
denominator of \eqref{Gab-finalnew}, and then we turn to the function
$N_{\lambda}(a,b)$ in the numerator.

\subsection{Holomorphic extension of $K(a,\lambda)$ in $\lambda$}
\label{sec:hol-lambda}

We study the holomorphicity of the map
$\C\ni \lambda\mapsto 
K(a,\lambda)+1+a
= \lambda W_{k(\lambda)}(\frac{1}{\lambda}
e^{\frac{1+a}{\lambda}})
$ for fixed $a\geq 0$. Our conventions for the branches $W_k$ 
indexed by $k\in \Z$ follow \cite[section~4]{Knuth}; see in particular 
figure~5 therein. In polar coordinates 
\begin{equation}
\lambda=|\lambda|e^{\iu \varphi}, \qquad \text{such that} \qquad
\frac{1}{\lambda}e^{(1+a)/\lambda}=
\frac{1}{|\lambda|}e^{\frac{1+a}{|\lambda|}\cos \varphi - 
\iu \big(\varphi+\frac{1+a}{|\lambda|}\sin \varphi\big) },
\label{eq:polar}%
\end{equation}
branch cuts of Lambert-W correspond to solutions of
$\varphi+\frac{1+a}{|\lambda|}\sin \varphi \in (2\Z+1)\pi$. 
\begin{Proposition}
\label{prop:branches}
Let $a\geq 0$ and $\lambda=|\lambda|e^{\iu \varphi} 
\in \C\setminus (-\infty,0]$. 
For $0 < \varphi <  \pi$, define 
by $\varphi+\frac{1+a}{\lambda_k(\varphi)}\sin \varphi =(2k+1)\pi$
a sequence  $(\lambda_k(\varphi))_{k\in \N_0}$ of positive 
real numbers, and let $\lambda_{-k-1}(-\varphi) \defas
\lambda_{k}(\varphi)$ as well as $\lambda_0(0)\defas 0$.
Then the assignment
\begin{align}
K(a,\lambda)
:=\left\{\!\!
\begin{array}{cl}
\lambda W_{-k}(\frac{1}{\lambda}e^{(1+a)/\lambda})-a-1 &
\text{if $\varphi>0$ and 
$\lambda_{k}(\varphi)<|\lambda| \leq \lambda_{k-1}(\varphi)$,}
\\
\lambda W_{0}(\frac{1}{\lambda}e^{(1+a)/\lambda})-a-1 &
\text{if $\varphi \geq 0$ and $|\lambda| >\lambda_{0}(\varphi)$,}
\\
\lambda W_{0}(\frac{1}{\lambda}e^{(1+a)/\lambda})-a-1 &
\text{if $\varphi < 0$ and $|\lambda| \geq \lambda_{-1}(\varphi)$,}
\\
\lambda W_{k}(\frac{1}{\lambda}e^{(1+a)/\lambda})-a-1 &
\text{if $\varphi < 0$ and 
$\lambda_{-k-1}(\varphi)\leq |\lambda| < \lambda_{-k}(\varphi)$}
\end{array}
\right.
\label{eq:W-branches}
\end{align}
uniquely extends to $\lambda<0$, and 
the resulting function $\lambda \mapsto K(a,\lambda)$ is 
holomorphic on $\C\setminus 
\{-(1+a) \frac{\sin \alpha}{\alpha} e^{\iu\alpha}\;:~
\frac{\sin \alpha}{\alpha} e^{\alpha \cot \alpha} \geq \frac{e}{1+a}\;,~
-\pi <\alpha <\pi\}$.
\end{Proposition}
\begin{proof}
  According to \cite{Knuth}, any branch $W_k$ is holomorphic on
  $\C\setminus (-\infty,0]$.  By \eqref{eq:polar} this means that
  $K(a,|\lambda|e^{\iu \varphi})$ is holomorpic in every open
  subregion of \eqref{eq:W-branches} excluding
  $|\lambda|=\lambda_k(\varphi)$.  First fix $0<\varphi<\pi$, $k\geq 1$ and
  consider $\lambda_\pm = |\lambda_\pm |e^{\iu \varphi}$ with
  $|\lambda_\pm|= \lambda_k(\varphi) \pm \epsilon$.  Then
  $\varphi+\frac{1+a}{|\lambda_\pm|}\sin \varphi = (2k+1)\mp \epsilon$
  and $\frac{1}{\lambda_\pm}e^{(1+a)/\lambda_\pm} =r_ke^{\mp
    \iu(\pi-\epsilon)}$ with $r_k:= 
  e^{\frac{(1+a)}{\lambda_k(\varphi)}\cos \varphi}/\lambda_k(\varphi) $. The assignment
  \eqref{eq:W-branches} gives $K(a,\lambda_+)=\lambda_+ W_{-k} (-r_k-
  \iu \epsilon)-a-1$ and $K(a,\lambda_-)= \lambda_- W_{-k-1} (-r_k+
  \iu \epsilon)-a-1$. The branch conventions of \cite{Knuth} imply
  $\lim_{\epsilon\to 0} K(a,\lambda_+)= \lim_{\epsilon\to 0}
  K(a,\lambda_-)$. The same continuity property holds for all $k\geq 1$ and 
$0<\varphi<0$. By the edge-of-the-wedge theorem, all
  branches $W_{-k}$ for $k\geq 1$ continue each other as the same
  holomorphic function.

Similarly, for $\varphi<0$ the passage from $W_{k}$ to 
$W_{k+1}$ is continuous for all $k\geq 1$, and  
all branches $W_{k}$ for $k\geq 1$ continue each other as the 
same holomorphic function. 

In terms of $\varphi=\pi-\alpha\in (0,\pi)$
or $\varphi=-\pi-\alpha \in (-\pi,0)$, the border 
between $W_0$ and $W_{\mp 1}$ is the curve 
$\mathcal{C}_a:= \{-(1+a) \frac{\sin \alpha}{\alpha} 
e^{\iu\alpha}\;:-\pi < \alpha < \pi\}$, when continued to 
$\alpha=0$.
The domain crossing is continuous if \eqref{eq:polar} falls into
$(-\infty, -\frac{1}{e}]$ at $\lambda_0(\varphi)e^{\pm \iu \varphi}$. 
This condition amounts to
$|\lambda|= (1+a)\frac{\sin \alpha}{\alpha} < r(\alpha)$ with 
$r(\alpha) \defas e^{1-\alpha \cot \alpha}$. It is satisfied for 
$|\alpha|\to \pi$, i.e.\ $\varphi\to 0$ and holds by continuity in an 
open interval $\varphi \in (-\varphi_c,\varphi_c)$. Again by the
edge-of-the-wedge theorem, the open domains of $W_0$ and 
$W_{\pm 1}$ continue each other as 
the same holomorphic function. 
The domain crossing at the part of $\mathcal{C}_a$ where
$(1+a)\frac{\sin \alpha}{\alpha}\geq e^{1-\alpha \cot \alpha}$ is 
not continuous and leads to a branch cut of $K(a,\lambda)$. 

The continuation to the negative real axis is determined by 
$\lim_{\varphi\to\pi} \lambda_0(\varphi)=1+a$ and
$\lim_{\varphi\to-\pi} \lambda_{-1}(\varphi)=1+a$. 
Therefore, \eqref{eq:W-branches} selects $W_0$ for 
$|\lambda|>1+a$ and $W_{\pm 1}$ for $|\lambda|<1+a$. 
According to
\eqref{eq:polar}, the argument of the Lambert function approaches
$-\frac{1}{|\lambda|} e^{-\frac{1}{|\lambda|}} 
e^{-\frac{a}{|\lambda|}} \in [-\frac{1}{e},0)$. For $t\in [-\frac{1}{e},0)$ one 
has $\lim_{\epsilon\to 0} W_0(t+\iu \epsilon)=
\lim_{\epsilon\to 0} W_0(t-\iu \epsilon)$ and 
$\lim_{\epsilon\to 0} W_{-1}(t+\iu \epsilon)=
\lim_{\epsilon\to 0} W_1(t-\iu \epsilon)$ so that the approach to the 
negative real axis is continuous, hence holomorphic. 
In standard conventions \cite{Knuth}, $\lambda \in (-1-a,0)$ is 
assigned to $W_{-1}$.

Finally, although every neighbourhood of $\lambda=0$ intersects 
all branches of
Lambert-W in $K(a,\lambda)$, we saw in \eqref{eq:K-FEQ} and
\eqref{eq:K-def} that $K(a,\lambda)=-\lambda\log(1+a) +
\bigO(\lambda^2)$ is holomorphic at $\lambda=0$. 
\end{proof}

We have seen that the domain boundaries of the Lambert function 
in $K(a,\lambda)$ are described by the curve
$\lambda_k(\varphi)e^{\iu \varphi}= 
(1+a)\frac{\sin((2k+1)\pi-\varphi)}{(2k+1)\pi-\varphi} 
e^{\iu \varphi}$. Taking the union over all $k$ and 
setting $\varphi-(2k+1)\mapsto \theta$, we obtain 
a part of the  \emph{cochleoid}\footnote{Cochleoid refers to 
`snail-shaped'. Its reciprocal $\frac{2b}{\pi} 
\frac{\theta}{\sin \theta} e^{\iu
  \theta}$ is the \emph{quadratrix of Hippias} used in
classical antiquity to trisect an angle or
to square a circle.}
$ \R \ni \theta \mapsto
	-(1+a) \frac{\sin \theta}{\theta} e^{\iu\theta}
$, see Figure~\ref{fig1}.
\begin{figure}[h]
\begin{picture}(60,72)
\put(0,0){\includegraphics[width=6cm]{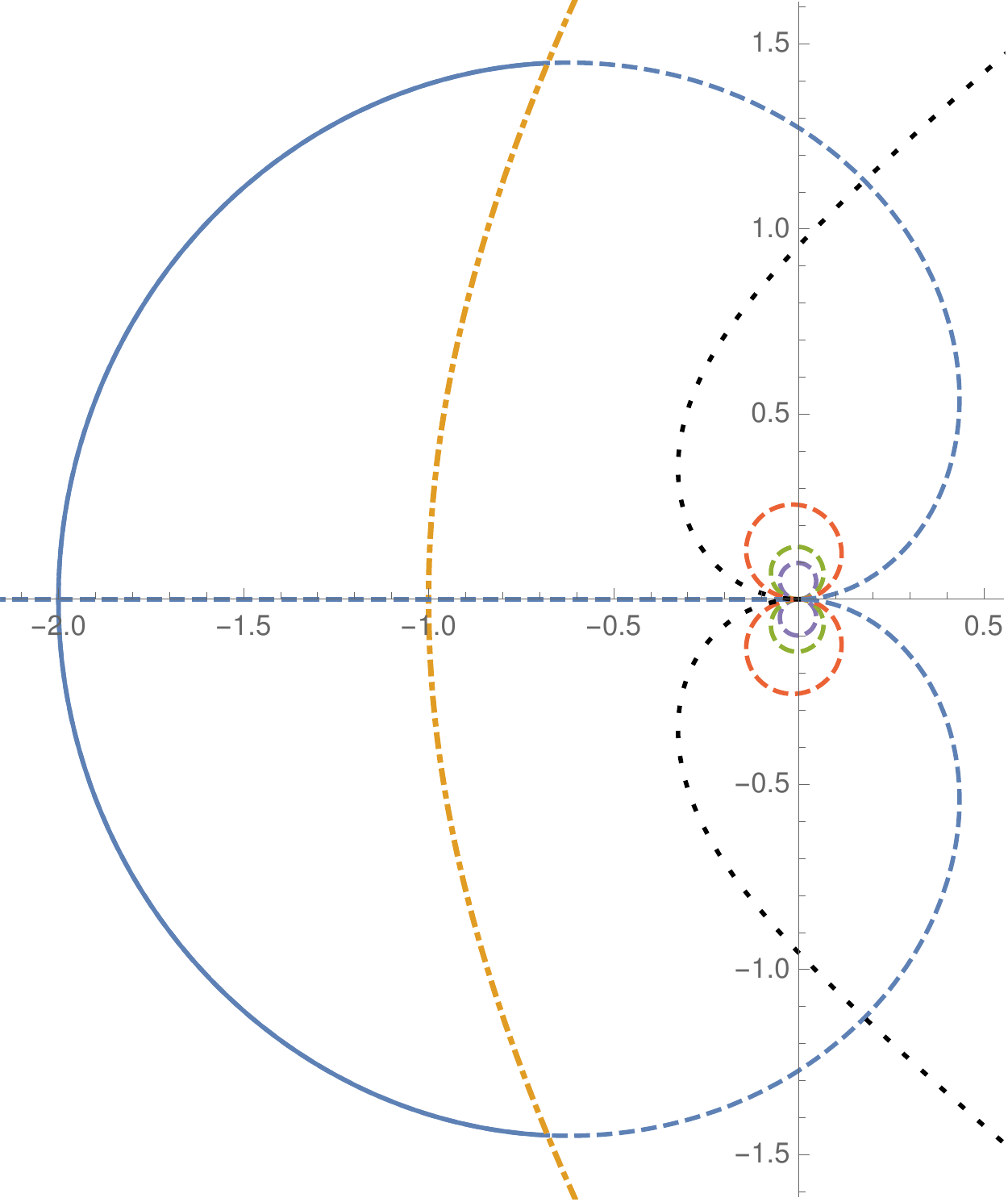}}
\put(65,50){\mbox{$W_0$}}
\put(65,20){\mbox{$W_0$}}
\put(-10,50){\mbox{$W_0$}}
\put(-10,20){\mbox{$W_0$}}
\put(33,55){\mbox{$W_{-1}$}}
\put(29.5,68){\mbox{$P$}}
\put(29,0){\mbox{$\bar{P}$}}
\put(33,15){\mbox{$W_1$}}
\put(23,25){\mbox{$\mathcal{C}$}}
\put(2,50){\mbox{$\mathcal{C}_a$}}
\put(60,68){\mbox{$\mathcal{L}_a^+$}}
\put(60,2){\mbox{$\mathcal{L}_a^-$}}
\put(50.5,24){\vector(-1,4){2}}
\put(50,21){\mbox{\scriptsize$W_2$}}
\put(50.5,47.5){\vector(-1,-4){2}}
\put(49,48){\mbox{\scriptsize$W_{{-}2}$}}
\put(62,34){\mbox{\tiny$\RePart(\lambda)$}}
\put(42,73){\mbox{\tiny$\ImPart(\lambda)$}}
\end{picture}
\caption{%
	Branch assignment of
	$
		\C\ni\lambda
		\mapsto 
		\lambda W_k(\frac{1}{\lambda}e^{(1+a)/\lambda})
	$ for $a=1$. 
	All dashed boundaries are holomorphically connected and form
        parts of a curve known as cochleoid. Only the solid part of
        the outer boundary $\mathcal{C}_a$ is a discontinuous branch cut.
	The conjugate pair of branch points $P$ and $\bar{P}$ traces 
the dash-dotted auxiliary curve $\mathcal{C}$ when we vary $a$. 
The region $\Omega_K$ to the right of $\mathcal{C}$ is the common
holomorphicity domain of $K(a,\lambda)$ and $L(a,\lambda)$ for all
$a\geq 0$.
Along the dotted curves $\mathcal{L}_a^\pm$ one has 
$\lambda W_l(\frac{1}{\lambda}e^{(1+a)/\lambda})<0$ for a different 
branch $W_l \neq  W_k$.
}%
\label{fig1}%
\end{figure}
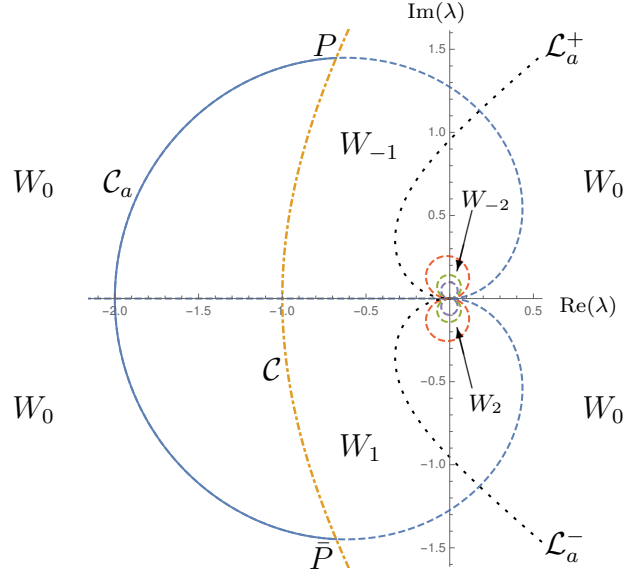
The cochleoid has infinitely many spirals
which all pass through $\lambda=0$ tangent to the real axis. The 
domain boundaries of assignments of branches of Lambert-W 
to $K(a,\lambda)$ correspond to every second spiral of the
cochleoid; the missing spirals
$\varphi+\frac{1+a}{\lambda}\sin \varphi =\pm 2k\pi$ correspond to the
centre lines of the branches.

The outer boundary 
$\mathcal{C}_a= \{-(1+a) \frac{\sin \alpha}{\alpha} 
e^{\iu\alpha}\;:-\pi < \alpha < \pi\}$
of the cochleoid intersects the circle 
of radius $r(\alpha)= e^{1-\alpha \cot \alpha}$ in the two branch points 
$-e^{\pm \iu \alpha}r(\alpha)$. Varying $a$, these branch points  
trace the critical curve 
\begin{align}
\mathcal{C}=\left\{ 
	-e^{1-\alpha \cot \alpha +\iu \alpha}
	\colon
	-\pi<\alpha< \pi
\right\}
\label{eq:critical}%
\end{align}
shown as dash-dotted curve in figure~\ref{fig1}.  
By construction, all branch cuts of 
$K(a,\lambda)$, when varying $a$, lie to the left of $\mathcal{C}$. 
Therefore, the region $\Omega_K$ to the
right of $\mathcal{C}$ is precisely the common domain of 
holomorphicity for all $a\geq 0$. In particular, the
domain $\Omega_K$ includes the disk $|\lambda|<1$, and 
$K(a,\lambda)$ has radius of convergence $\geq 1$ in $\lambda$, 
for all values $a \geq 0$.

Observe that $\R_+$ (that is $\alpha=\pm\pi$) is the only infinite ray
in the $\lambda$-plane along which $K(a,\lambda)$ may be analytically
continued for all values of $a\geq 0$. Any other ray
$-e^{\iu\alpha}\R_+$ where $\abs{\alpha}<\pi$, hits a point on
$\mathcal{C}$.

\begin{Remark}[Holomorphicity domain]
The maximal domain of analytic continuation of $K(a,\lambda)$
into the complex $\lambda$-plane, for all $a\geq 0$, can also be derived directly from \eqref{eq:K-FEQ}.
A branch point of $K(a,\lambda)$ at $\lambda=\lambda^{\ast}$
corresponds to a zero at $K=K^{\ast} \defas K(a,\lambda^{\ast})$ of the derivative
\begin{equation*}
0 \stackrel{!}{=}
\frac{\td \lambda}{\td K}
= \frac{\td}{\td K} \frac{-K}{\log(1+a+K)}
= \frac{K-(1+a+K)\log(1+a+K)}{(1+a+K) \log^2(1+a+K)}.
\end{equation*}
Inserting \eqref{eq:K-FEQ} and \eqref{eq:K-Lambert} gives 
$W(\frac{1}{\lambda^*}e^{\frac{1+a}{\lambda^*}})=-1$ with solution
$\frac{1+a}{\lambda^*}= W_0(-\frac{1+a}{e} \pm \iu
\epsilon)$. This equation can be converted into 
the same critical curve \eqref{eq:critical}.
\end{Remark}
\begin{Remark}[Symmetry]\label{rem:symmetry}
	It follows from our above discussion and $W_k(\bar{z})=\overline{W_{-k}(z)}$ that $K(a,\bar{\lambda}) = \overline{K(a,\lambda)}$ for every $\lambda \in \Omega_K$ and all real $a\geq 0$.
\end{Remark}
\begin{Lemma}
	The function $L(a,\lambda) = \log\frac{a+K(a,\lambda)}{a}$ from \eqref{eq:L-Lambert}, where $\log$ denotes the principal branch, is for every $a> 0$ holomorphic in $\lambda$ on the domain $\Omega_K \ni \lambda$.
\end{Lemma}
\begin{proof}
	First note that $K(a,\lambda)=-1-a$ requires $\lambda=0$ in \eqref{eq:K-FEQ} and contradicts $K(a,0)=0$. For positive values of $\lambda$, the monotonicity of the branch $W_0$ implies
	\begin{equation*}
		K(a,\lambda)
		=\lambda W_0\left( \frac{e^{1/\lambda}}{\lambda} e^{a/\lambda} \right)-1-a
		\geq
		\lambda W_0\left( \frac{e^{1/\lambda}}{\lambda} \right)-1-a
		= -a
	\end{equation*}
	with equality if and only if $a=0$. The same holds for $\lambda \in(-1,0)$ and we conclude
	\begin{equation}
		a+K(a,\lambda)>0
		\quad\text{for all real $\lambda>-1$ and any $a>0$}.
		\label{eq:a+K>0}%
	\end{equation}
Now observe that \eqref{eq:K-FEQ} admits a solution
\[
\lambda^\pm_a(\varphi)\; W_k\Big( \frac{1}{\lambda^\pm_a(\varphi)}
e^{\frac{1+a}{\lambda^\pm_a(\varphi)}}\Big)
=-e^{\pi \cot \varphi} \mp \iu \epsilon 
{}~~ \Leftrightarrow ~~
\lambda^\pm_a(\varphi) \equiv 
\frac{\sin \varphi}{\pi}(1{+}a{+}e^{\pi \cot \varphi}) 
e^{\pm \iu \varphi},
\]
for $0<\varphi<\pi$, where indeed 
$\log(\lambda W_k( \frac{1}{\lambda}e^{(1+a)/\lambda}))$ is not 
defined for \emph{some} branch $W_k$. 
According to Proposition \ref{prop:branches}, the critical curve 
$\mathcal{L}_a^+:=\{ \lambda^+_a(\varphi)\;:~0<\varphi\leq \pi\}$ 
crosses domains of the Lambert function at  
$(2k+1)\pi=\varphi_k +\pi \frac{1+a}{1+a+e^{\pi \cot \varphi_k}}$. This equation
admits only a single solution $\varphi_0\in (0,\pi)$.
Similarly, its mirror 
$\mathcal{L}_a^-:=\{ \lambda^-_a(\varphi)\;:~0<\varphi\leq
\pi\}$ crosses domains of the Lambert function at the same 
angle $\varphi_0\in (0,\pi)$. See also Figure~\ref{fig1}.
Then \eqref{eq:W-branches} implies 
\begin{itemize}
\item 
$K(a,\lambda):= \lambda W_0 (\frac{1}{\lambda}e^{(1+a)/\lambda})-1-a$ 
for $\lambda\equiv\lambda_a^\pm (\varphi)$ and $0<\varphi<\varphi_0$,
 and 
\item $K(a,\lambda) :=\lambda W_{\mp 1} 
(\frac{1}{\lambda}e^{(1+a)/\lambda})-1-a$ 
for $\lambda\equiv\lambda_a^\pm (\varphi)$ and $\varphi_0<\varphi<\pi$.
\end{itemize}
It turns out that the spurious solution is precisely the opposite branch assignment: 
\begin{itemize}
\item $\lambda W_{\pm 1} (\frac{1}{\lambda}e^{(1+a)/\lambda})=-e^{\pi \cot \varphi}$
for $\lambda\equiv\lambda_a^\pm (\varphi)$ and $0<\varphi<\varphi_0$, and 
\item $\lambda W_{0} 
(\frac{1}{\lambda}e^{(1+a)/\lambda})=-e^{\pi \cot \varphi}$ 
for $\lambda\equiv\lambda_a^\pm (\varphi)$ and $\varphi_0<\varphi<\pi$.
\end{itemize}
Indeed, this assignment can be established in the limits $\varphi\to 0$ and 
$\varphi\to \pi$. By continuity inside every open branch domain, 
this assignment can only switch at the domain 
crossing $\varphi=\varphi_0$. 

In conclusion, the branch assignment \eqref{eq:W-branches} implies
$\mathrm{arg}(K(a,\lambda)+1+a) \in (-\pi,\pi)$ for all $\lambda\in
\Omega_K$ (at $\lambda^\pm_a(\varphi_0)$ by the mean value property of
holomorphic functions). 
Therefore, $\lambda=-K/\log(1+a+K)$ is a well-defined, single-valued 
inverse function on the image of $K$. In particular, for any 
fixed $a> 0$, the map $\lambda \mapsto K(a,\lambda)$ is a 
biholomorphic bijection of $\Omega_K$ onto its image. Similarly, 
for fixed $0\neq\lambda\in \Omega_K$, the map $a\mapsto K(a,\lambda)$ 
is injective with inverse $a=e^{-K/\lambda}-1-K$.

For the same reasons, also $\mathrm{arg}(K(a,\lambda)+a)
\in (-\pi,\pi)$, and in consequence $L(a,\lambda)$ is well-defined for
all $a>0$ and $\lambda \in \Omega_K$. In the limit $a\rightarrow 0$,
we find $\frac{a+K(a,\lambda)}{a} \rightarrow \frac{1}{1+\lambda}$ and
thus $L(a,\lambda)\rightarrow - \log(1+\lambda)$.
\end{proof}

\subsection{Complexification of $a$ and $b$ for real $\lambda$}
We discuss the cases of positive and negative $\lambda$ separately.

\begin{Lemma}
\label{Lemma:K-pos}
Let $\lambda>0$ and 
$B_\lambda^\pm \defas \{\pm \lambda\pi
\iu-t\colon 1+\lambda-\lambda\log\lambda\leq t<\infty\}$.
Then
\begin{equation} 
K(z,\lambda)\defas \lambda \,W_k\Big(
\frac{1}{\lambda}
e^{\frac{1+z}{\lambda}}
	\Big)-1-z\quad
\text{if } (2k-1)\pi\lambda  <
\ImPart(z) \leq (2k+1)\pi\lambda
\label{eq:K-Lambert-z}%
\end{equation}
is holomorphic on 
$z \in \C\setminus (B^+_\lambda \cup B^-_\lambda)$.
\end{Lemma}
\begin{proof}
The logarithm in \eqref{eq:K-FEQ} is well-defined if $
W(\frac{1}{\lambda}e^{\frac{1+z}{\lambda}}) \notin (-\infty,0]$.
On a domain of $z$ satisfying this condition we thus have 
$\ImPart(\frac{1+z}{\lambda}- W(
\frac{1}{\lambda}e^{\frac{1+z}{\lambda}})) \in (-\pi,\pi)$. The 
branch assignment \eqref{eq:K-Lambert-z} is then a consequence 
of the branch conventions \cite{Knuth} of the Lambert
function. These branches are holomorphic in every open 
strip $(2k-1)\pi\lambda  < \ImPart(z) < (2k+1)\pi\lambda$.

Only $W_{-1}(w)$ and $W_0(w)$ reach negative reals \cite{Knuth},
for $w \in [-\frac{1}{e},0]$. Together with the branch assignment
\eqref{eq:K-Lambert-z} it follows that
$
W_0(\frac{1}{\lambda}e^{\frac{1+z}{\lambda}}) \in [-1,0]$
iff $z\in B_\lambda^+$ and
$
W_{-1}(\frac{1}{\lambda}e^{\frac{1+z}{\lambda}}) \in (-\infty,-1]$
iff $z\in B_\lambda^-$. In the same way as in 
Proposition~\ref{prop:branches} one proves that all other domain crossings 
from $W_k$ to $W_{k+1}$ at $\ImPart(z) = (2k+1)\pi\lambda$
are continuous, including the border between $W_0$ and
$W_{\pm 1}$ at $(\R\pm\lambda\pi \iu)\setminus B^\pm_ \lambda$. By the
edge-of-the-wedge theorem, $\C\setminus (B^+_\lambda
\cup B^-_\lambda) \ni z \mapsto K(z,\lambda)$ is holomorphic for any 
given $\lambda>0$.
\end{proof}

\begin{Lemma}
\label{Lemma:K-neg}
Let $-1<\lambda<0$ and set $B^0_{\lambda} \defas (-\infty,-1+|\lambda|-|\lambda|\log|\lambda|)$. Then 
\begin{equation} 
K(z,\lambda) \defas \lambda W_{-k}\Big(
\frac{1}{\lambda}
 e^{\frac{1+z}{\lambda}}\Big)-1-z\quad
{
	\text{if } 0\leq (2k-2)\pi\abs{\lambda} \leq \ImPart(z) < 2k\pi\abs{\lambda}
	\atop
	\text{or } 2k\pi\abs{\lambda}  \leq \ImPart(z) 
< (2k+2)\pi\abs{\lambda}\leq 0
}
\label{eq:K-Lambert-zneg}%
 \end{equation}
for $k\in \Z\setminus\{0\}$ is holomorphic on $z \in \C\setminus B_{\lambda}^0$.
\end{Lemma}
\begin{proof}
We know from sec.~\ref{sec:hol-lambda} that $K$ selects 
for $z=a\in \R_+$ the branch $W_{-1}$. Equations 
\eqref{eq:K-FEQ} and \eqref{eq:K-Lambert} combine 
to $\ImPart(\frac{1+z}{|\lambda|}+ W(
-\frac{1}{|\lambda|}e^{-(1+z)/|\lambda|})) \in (-\pi,\pi)$. The 
branch assignment \eqref{eq:K-Lambert-zneg} is then a consequence 
of the branch conventions \cite{Knuth} of the Lambert
function. These branches are holomorphic in every open 
strip $2k\pi|\lambda|  <
\ImPart(z) < (2k+2)\pi|\lambda|$ and continuously glue to each other
except when $W_{-1}(w)$ and $W_{1}(w)$ are separated by $W_0$.
These are the points $w=-\frac{1}{|\lambda|}e^{-(1+z)/|\lambda|}
\leq -\frac{1}{e}$ which correspond to $z \in B_{\lambda}^0$. 
\end{proof}

\subsection*{Summary}
The function $I_{\lambda}(a)=K(a,\lambda)-\lambda L(a,\lambda)$ can be
extended precisely to $\lambda \in \Omega_K$ (the region right 
of $\mathcal{C}$), if it is to be defined for all $a\geq 0$. This 
domain contains the real interval $(-1,\infty)$.
For $\lambda\geq 0$ the branch $W_0$ is selected at
$a\geq 0$; it extends to any $z\in \C\setminus (B^+_\lambda\cup
B^-_\lambda)$ according to Lemma \ref{Lemma:K-pos}. 
For $-1<\lambda< 0$ the branch $W_{-1}$ is selected at
$a\geq 0$; it extends to any $z\in \C\setminus B^0$ according to 
Lemma \ref{Lemma:K-neg}.

\begin{Remark}[Strong coupling]
 	\label{rem:large-lambda}%
 	When $\lambda\rightarrow \infty$, eventually $z=\frac{1}{\lambda} e^{(1+a)/\lambda}$ fulfils $\abs{z}<1/e$ and thus the series $W_0(z)=\sum_{n=1}^{\infty} \frac{(-n)^{n-1}}{n!} z^n$ yields a convergent strong-coupling expansion of $I_{\lambda}(a)$. The condition $\abs{z}<1/e$ is then equivalent to
 	\begin{equation*}
 		\lambda > \frac{1+a}{W_0\left( \frac{1+a}{e} \right)}
 		.
 	\end{equation*}
\end{Remark}

\subsection{Holomorphic extension of $N_\lambda(a,b)$}

The domain of holomorphicity in $\lambda$ common to all $a \geq 0$ is 
according to \eqref{eq:Nlambda} the region to the right of the 
envelope $\mathcal{E}$ of complex rays 
\begin{align*}
\{P(t)+a\cdot m(t)\colon t\in \R\}_{a\in \R_+}\;,\quad
P(t) \defas \frac{\frac{1}{2}+\iu t}{\log(\frac{1}{2}-\iu t)} \;,~
m(t) \defas \frac{1}{\log(\frac{1}{2}-\iu t)} \;.
\end{align*}
One finds the following parametrisation for the envelope:
\begin{equation*}
\mathcal{E}(t)= \begin{cases}
	P(t) 
& \text{for } 0\leq |t|\leq t_{\mathcal{E}} \;, \\[0.8ex]
	P(t)+ m(t) \frac{\overline{m(t)} P'(t)-m(t) \overline{P'(t)}}{
	m(t) \overline{m'(t)} -\overline{m(t)} m'(t)}
& \text{for } |t|\geq t_{\mathcal{E}} \;,
\end{cases}
\end{equation*}
where $t_{\mathcal{E}}$ is defined by  
$\ImPart\big(\overline{m(t_{\mathcal{E}})} P'(t_{\mathcal{E}})\big)=0$. 
Setting $t_{\mathcal{E}}= \frac{1}{2} \tan \psi$, this condition
amounts to $\psi^2 + (\log (2\cos \psi))^2 - \psi \sin(2\psi) - 
 \cos(2\psi) \log(2\cos \psi)=0$, evaluated to $\psi\approx 0.861$ and 
$t_{\mathcal{E}} \approx 0.582$. See figure~\ref{fig2}.
\begin{figure}[h]
\begin{picture}(60,52)
\put(-10,-45){\includegraphics[width=8cm]{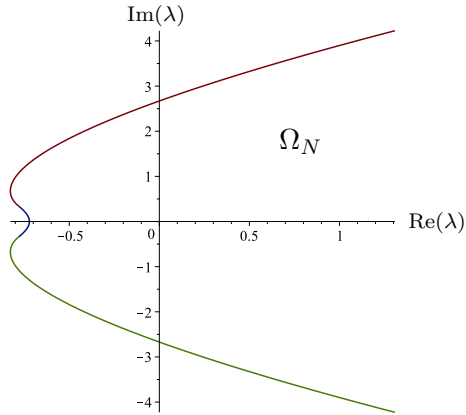}}
\put(51,24.5){\mbox{\tiny$\RePart(\lambda)$}}
\put(14,52){\mbox{\tiny$\ImPart(\lambda)$}}
\put(34,35){\mbox{\small$\Omega_N$}}
\end{picture}
\caption{\label{fig2}The envelope $\mathcal{E}$ of rays parametrised by $a\in \mathbb{R}_+$. The blue part is traced by $a=0$, red and green parts 
by $a>0$. The domain $\Omega_N$ to the right of $\mathcal{E}$ is 
the joint holomorphicity domain of
$N_\lambda(a,b)$ and $G_\lambda(a,b)$ for all $a,b\in \R_+$.}
\end{figure}
The joint domain $\Omega_N$ of holomorphicity
of $N_\lambda(a,b)$ for all $a,b\in \R_+$ is the region to the right
of $\mathcal{E}$. One finds $\Omega_N \subset \Omega_K$ so that 
$\Omega_N$ is also the joint holomorphicity domain of
$\lambda \mapsto G_\lambda(a,b)$ for all $a,b\in \R_+$. This domain
contains the real interval $(-\frac{1}{\log 4},\infty)$. In
particular, the perturbation series of $G_\lambda(a,b)$ has radius of
convergence $\frac{1}{\log 4}$.

The complexification of $N_\lambda(a,b)$ at fixed real
$\lambda >-1/\log 4$ to $a\mapsto z\in \C$ and $b\mapsto w\in \C$ is
according to \eqref{eq:Nlambda} holomorphic on a joint domain
$\Omega_\lambda$ which is the intersection of the half space
$\RePart(z)>-\frac{1}{2}$ with the domain to the right of the critical
curve $\mathcal{N}_\lambda \defas \{-\frac{1}{2}+\iu t+\lambda
\log(\frac{1}{2}+\iu t)\colon t\in \R\}$.

\subsection{A warning}
\label{sec:warning}

Holomorphicity of the final solution $G_\lambda(a,b)$ in
$(-\frac{1}{\log 4},\infty)$ does \emph{not} mean that intermediate
purely real formulae are analytic near $\lambda=0$. First, positivity
of \eqref{Gab-real} requires to take for $\lambda<0$ the
$\arctan$-branch in $[-\pi,0]$.  However, \eqref{eq:J1} and
\eqref{eq:J2} \emph{do not} extend to $-1<\lambda<0$ by merely
replacing $W_0$ by $W_{-1}$ and
$\arctan_{[0,\pi]}$ with $\arctan_{[-\pi,0]}$.  Instead, we get
\begin{Lemma}
\label{Lemma:Jneg}
For any $a\geq 0$ and $-1<\lambda < 0$ and $z \in \C\setminus
{(-\infty,-1]}$ one has
\begin{gather*}
\int_0^\infty \! \frac{\td u}{\pi}\; 
\frac{\mbox{\small{$\displaystyle 
\arctan\displaylimits_{[-\pi,0]} 
\Big(\frac{\lambda\pi}{1+a+u -\lambda \log u}\Big)$}}}{
1+u+z}
= 
\log \left( \frac{(z+\lambda \log(1{+}z)-a)(1+z)}{
\prod_{k=-1}^0 (1{+}z{-}\lambda
  W_{k}(\frac{1}{\lambda} e^{(1+a)/\lambda}))}\right),
\\
\int_0^\infty \! \frac{\td u}{\pi}\; 
\Big(\!
\arctan\displaylimits_{[-\pi,0]}  
\Big(\frac{\lambda\pi}{1{+}a{+}u {-}\lambda \log u}\Big)
-\frac{\lambda\pi}{1{+}u}\Big)
= -1-a+\sum_{k=-1}^0 
\lambda W_k\Big(\frac{1}{\lambda} e^{(1+a)/\lambda}\Big) .
\end{gather*}
\end{Lemma}

\begin{proof} As in the proof of Lemma~\ref{Lemma:J} the integrals are
  rewritten as contour integral of the complex logarithm
$\log (1-\frac{\lambda \log (-w)}{1+a+w})$. However, the branch cut is
not only $\R_+$ via $\log(-w)$ but extends to the intervals 
$[-\lambda W_0(\frac{1}{\lambda}e^{(1+a)/\lambda}),0]$ and 
$(-\infty, -\lambda W_{-1}(\frac{1}{\lambda}e^{(1+a)/\lambda})]$ 
on the negative real axis. Therefore, the contour must encircle 
the extended interval 
$[-\lambda W_0(\frac{1}{\lambda}e^{(1+a)/\lambda}),\infty)$. The jump
$2\pi \iu$ of $\log (1-\frac{\lambda \log (-w)}{1+a+w})$ along
$[-\lambda W_0(\frac{1}{\lambda}e^{(1+a)/\lambda}),0]$ contributes the
additional terms compared with the na\"{\i}ve analytic continuation.
\end{proof}
The analogue of Proposition~\ref{prop:Ilambda} is
\begin{Corollary}
	For any $a>0$ and $-1<\lambda<0$ one has 
\begin{align}
&\int_0^\infty  \frac{\td p}{\pi} \bigg( \!
\arctan\displaylimits_{[-\pi,0]} \!\bigg(\frac{\lambda\pi}{a+ 
	\lambda \,W_{-1}\big(\frac{1}{\lambda} e^{\frac{1+p}{\lambda}}\big)
	-\lambda \log \big( \lambda W_{-1}\big( \frac{1}{\lambda} 
        e^{\frac{1+p}{\lambda}} \big)-1\big)}\bigg)
-\frac{\lambda\pi}{1{+}p}\bigg) 
\nonumber
\\
&=  -1-a + \lambda \log a + 
\lambda \,W_{-1}\Big(\frac{1}{\lambda}e^{\frac{1+a}{\lambda}}\Big)
-\lambda \log \Big(\lambda W_{-1}\Big( \frac{1}{\lambda}
        e^{\frac{1+a}{\lambda}}\Big){-}1\Big)
\\
& + \lambda \,W_0\Big(\frac{1}{\lambda}e^{\frac{1+a}{\lambda}}\Big)
-\lambda \log \Big(1-\lambda W_0\Big( \frac{1}{\lambda}
        e^{\frac{1+a}{\lambda}}\Big)\Big).
\nonumber
\end{align}
\end{Corollary}
This means that the real-analytic continuation of $\tau_a(p)$ does \emph{not} 
solve \eqref{cottauba} for $\lambda<0$. This equation can be rescued
when including a term $\frac{\CO}{\lambda \pi(\CO-b)}h_\lambda(a)$
with 
\[
h_\lambda(a)=\left\{ \begin{array}{cl} 0 &\text{ for }\lambda\geq 0 \\
-\lambda \,W_0\Big(\frac{1}{\lambda}e^{\frac{1+a}{\lambda}}\Big)
+\lambda \log \Big(\lambda W_0\Big( \frac{1}{\lambda}
         e^{\frac{1+a}{\lambda}}\Big){-}1\Big) &
\text{ for }\lambda<0 
\end{array}\right.
\]
on the rhs of \eqref{cottauba-co}. It corresponds to a 
non-trivial solution of the homogeneous Carleman
equation as discussed after \eqref{eq:Carl-hom}. Note that 
$h_\lambda(a)$ is smooth but not analytic in $\lambda=0$. Its Taylor series in
$\lambda=0$ vanishes identically so that it does not interfere with the
perturbative solution of sec.\ \ref{sec:perturbative}.

For the same reasons, the purely real formula
\eqref{Gab-alternative} does not generalise to $\lambda<0$.

\section{Perturbative expansion of $N_\lambda(a,b)$}
\label{sec:finalintegral}

The $\lambda$-expansion of $N_{\lambda}(a,b)$ from \eqref{eq:Nlambda-1} is computable symbolically with {\HyperInt}:
\begin{MapleInput}
ln(1-lambda*ln(-z)/(1+a+z))*ln(1-lambda*ln(1+z+w)/(b-z-w));
eval(diff(%,w),w=0);
series(%,lambda):
coeff(%,lambda,2);
hyperInt(%/(2*I*Pi),z);
X := fibrationBasis(%,[a,b]);
\end{MapleInput}
The integral $X$ over $z$ from $0$ to $\infty$ depends on the signs
$\delta_z$ and $\delta_b$ of the imaginary parts of $z$ and $b$. To
get the half of $\gamma_{\epsilon}^+$ above $\R_+$, we have to set
$\delta_z=1$ and $\delta_b=-1$ (since $b$ lies below $z$). The first
half $\int_{\infty}^0 = - \int_0^{\infty}$ of $\gamma_{\epsilon}^+$
subtracts the conjugate, so
\begin{MapleInput}
Above := eval(X,[delta[z]=1,delta[b]=-1]);
Below := eval(X,[delta[z]=-1,delta[b]=1]);
N[2] := collect(Above-Below,Hlog,factor);
\end{MapleInput}
\begin{MapleMath}
	-\frac{
		\Hlog{a}{-1}
	}{a(1+a+b)}
	-\frac{
		\Hlog{b}{-1}
	}{b(1+a+b)}
	+\frac{
		\Hlog{a}{0, -1}
		+\Hlog{b}{0, -1}
		+\zeta(2)
	}{
		(1+a+b)^2
	} 
\end{MapleMath}
gives $[\lambda^2] N_{\lambda}(a,b) = \frac{\zeta(2)-\Li_2(-a)-\Li_2(-b)}{(1+a+b)^2}-\frac{\log(1+a)}{a(1+a+b)}-\frac{\log(1+b)}{b(1+a+b)}$. These terms 
show up in \eqref{Gab-pert}. 
In fact, we can characterize the emerging polylogarithms very precisely: They belong to the family of multiple polylogarithms studied by Nielsen, \cite{Nielsen:DilogarithmusVerallgemeinerungen,Koelbig:Nielsen},
\begin{equation}
	S_{n,p}(z) = \frac{(-1)^{n+p-1}}{(n-1)! p!}
	\int_0^1 \frac{\log^{n-1}(t) \log^p(1-zt)}{t}\ \td t
	.
	\label{eq:Nielsen}%
\end{equation}
To make this clear, we expand \eqref{eq:Nlambda-1} in the form
\begin{align}
N_\lambda(a,b)=
\frac{\partial}{\partial w}
\sum_{m,n=1}^\infty \! \frac{(-\lambda)^{m+n}}{m!n!}
\frac{\partial^{m-1}}{\partial a^{m-1}}
\frac{\partial^{n-1}}{\partial b^{n-1}}
\int_{\gamma_\epsilon^+} \!
\frac{\td z}{2\pi \iu} 
\frac{(\log(-z))^m(\log(1{+}z{+}w))^n}{(1+a+z)(b-z-w)},
\label{eq:Nloglog}
\end{align}
taken at $w=0$. Pulling out the prefactor $(1+a+b-w)^{-1}$, the decomposition
\begin{equation*}
	\frac{1+a+b-w}{(1+a+z)(b-z-w)}
	= 
		\frac{1+b}{(b-z-w)(1+w+z)}
		-\frac{a-w}{(1+a+z)(1+w+z)} 
\end{equation*}
of the integration kernel completely separates the $a$- and $b$ dependence of the integral (up to the prefactor). The remaining integrals can be transformed into the form \eqref{eq:Nielsen}. In fact, we can compute the generating function
\begin{align}
R_{\alpha,\beta}(a,b;w)
&\defas \frac{1}{\pi} \int_0^\infty 
\td z \; \ImPart \left( 
\frac{(-z-\iu \epsilon)^\alpha (1+z+w+\iu \epsilon)^\beta}{
(1+a+z+\iu \epsilon)(b-z-w-\iu \epsilon)}
\right)
\label{eq:R-diff}
\end{align}
defined so that the coefficient of $\alpha^m \beta^n$ is $m! n!$ 
times the contour integral in \eqref{eq:Nloglog}:
\begin{equation}
	N_{\lambda}(a,b)
	= \sum_{m,n=1}^{\infty} \partial_a^{m-1} \partial_b^{n-1}
	[ \alpha^m \beta^n] \left. \frac{\partial}{\partial w}\right|_{w=0}
	R_{-\lambda \alpha,-\lambda\beta}(a,b;w)
	.
	\label{eq:N-from-R}%
\end{equation}
Resolving the $\iu \epsilon$-descriptions in \eqref{eq:R-diff} 
gives after substitution $z=(1+w)p$
\begin{align}
&R_{\alpha,\beta}(a,b;w)
\nonumber
\\
&=\cos(\alpha \pi)
\frac{(b-w)^\alpha (1+b)^\beta}{(1+a+b-w)}
+\frac{(1+w)^{\alpha+\beta-1}\sin (\pi \alpha)}{\pi(1+a+b-w)} 
(a-w) \int_0^\infty \td p \; 
\frac{p^\alpha (1+p)^{\beta-1}}{(p+\frac{1+a}{1+w})}
\nonumber
\\
&+\frac{(1+w)^{\alpha+\beta-1}\sin (\pi \alpha)}{\pi(1+a+b-w)} 
(1+b) \intbar_0^\infty 
\td p \; 
\frac{p^\alpha (1+p)^{\beta-1}}{(p-\frac{(b-w)}{1+w})}\;.
\label{eq:Rabw}
\end{align}
The integral in the 2nd line is a standard hypergeometric integral. 
Writing $p^\alpha (1+p)^{\beta-1}=\frac{1}{\Gamma(1-\beta)}
G^{1,1}_{1,1}\big(p|\genfrac{}{}{0pt}{}{\alpha+\beta}{\alpha}\big)$, the 
last line of \eqref{eq:Rabw} is the Hilbert transform at
$\frac{b-w}{1+w}$ of a Meijer-G function. This Hilbert transform 
is simply obtained by 
adding a leading 0 and a terminating $\frac{1}{2}$ to both rows of
arguments:
\begin{align*}
\frac{1}{\pi} \intbar_0^\infty \td p \; 
\frac{p^\alpha (1+p)^{\beta-1}}{(p-\frac{(b-w)}{1+w})}
&= \frac{1}{\Gamma(1-\beta)}
G^{2,2}_{3,3}\Big(\frac{b-w}{1+w}\Big| 
\genfrac{}{}{0pt}{}{0,\alpha+\beta;\frac{1}{2}}{
0,\alpha;\frac{1}{2}}\Big)\;.
\end{align*}
The Meijer-G function on the rhs is expanded into a ${}_2F_1$ function and a 
 ${}_1F_0$ function. The latter one cancels the first term on the rhs of 
\eqref{eq:Rabw}, giving 
\begin{align*}
R_{\alpha,\beta}(a,b;w)
&= \frac{(1+w)^{\alpha+\beta}}{(1{+}a{+}b{-}w)} 
\frac{(a-w)}{(1+a)}
\frac{\alpha\Gamma(1{-}\alpha{-}\beta)}{
\Gamma(2{-}\beta)\Gamma(1{-}\alpha)}
{}_2F_1\Big(\genfrac{}{}{0pt}{}{1,1{+}\alpha}{2{-}\beta}\Big| 
\frac{a-w}{1+a}\Big) 
\\
&+\frac{(1+w)^{\alpha+\beta}}{(1{+}a{+}b{-}w)} 
\frac{\Gamma(1{-}\alpha{-}\beta)}{\Gamma(1{-}\alpha)\Gamma(1{-}\beta)}
\frac{1+b}{1+w}
{}_2F_1\Big(\genfrac{}{}{0pt}{}{1,1{-}\alpha{-}\beta}{1{-}\alpha}\Big| 
\frac{w-b}{1+w}\Big) \;.
\end{align*}
A contiguous relation in the first line together with 
fractional transformations of both lines bring this formula 
into the following manifestly symmetric form:
\begin{equation}
\begin{split}
R_{\alpha,\beta}(a,b;w)
&= \frac{1}{(1+a+b-w)} 
\frac{\Gamma(1-\alpha-\beta)}{
\Gamma(1-\alpha)\Gamma(1-\beta)}
\bigg\{
-(1+w)^{\alpha+\beta}
\\ & \qquad \qquad
+(1+w)^{\beta}(1+a)^\alpha
{}_2F_1\Big(\genfrac{}{}{0pt}{}{-\beta,\alpha}{1-\beta}\Big| 
\frac{w-a}{1+w}\Big) 
\\ &\qquad \qquad
+ (1+w)^\alpha(1+b)^\beta 
{}_2F_1\Big(\genfrac{}{}{0pt}{}{-\alpha,\beta}{1-\alpha}\Big| 
\frac{w-b}{1+w}\Big) 
\bigg\}\;.
\label{eq:R-2F1}%
\end{split}
\end{equation}
This hypergeometric function generates the Nielsen polylogarithms, as observed in \cite[Equation~(2.12)]{KoelbigMignacoRemiddi:NielsenNumerical} and \cite[Theorem~6.6]{BorweinBradleyBroadhurstLisonek:SpecialValues}:
\begin{equation}
	{_2F_1} \left( -x, y \atop 1-x \middle| z \right)
	= 1-\sum_{n,p \geq 1} S_{n,p}(z) x^n y^p 
	.
	\label{eq:Nielsen-2F1}%
\end{equation}
We note that the Gamma functions in \eqref{eq:R-2F1} expand into Riemann 
zeta values,
\begin{equation}
\frac{\Gamma(1-\alpha-\beta)}{
\Gamma(1-\alpha)\Gamma(1-\beta)}
= \exp\Big( \sum_{k=2}^\infty 
\big((\alpha+\beta)^k-\alpha^k-\beta^k \big) \frac{\zeta(k)}{k}\Big)\;.
\end{equation}
For the derivative with respect to $w$, note that for $z=(w-a)/(1+w)$, we have
\begin{equation*}
	\left.\frac{\partial}{\partial w} \right|_{w=0} S_{n,p}(z)
	= S_{n,p}'(-a)
	\left.\frac{\partial z}{\partial w} \right|_{w=0} 
	= (1+a) S_{n,p}'(-a)
	= -(1+a) \partial_a S_{n,p}(-a)
	.
\end{equation*}
The contribution to $\partial_w|_{w=0} R$ from the second line 
in \eqref{eq:R-2F1} is then
\begin{equation*}
	\frac{(1+a)^{\alpha}}{1+a+b}
	\left\{ \frac{1}{1+a+b} + \beta - (1+a)\partial_a \right\}
	\left( 1-\sum_{n,p=1}^{\infty} S_{n,p}(-a) \beta^n \alpha^p \right),
\end{equation*}
up to the Gamma prefactor. Hence, we can compute the expansion of
$\partial_w|_{w=0} R$ in terms of zeta values, logarithms $\log(1+a)$
and polylogarithms $S_{n,p}(-a)$ (and those with $a$ replaced by $b$),
with rational functions of $a$ and $b$ as coefficients.

\begin{Remark}
The same strategy applied to \eqref{eq:Nlambda} leads with the integral
representation 
$\frac{1}{2\pi}\int_0^\infty \td t \;(\frac{1}{2}+\iu t)^{-x} 
(\frac{1}{2}-\iu t)^{-y}= \frac{1}{(x+y-1)B(x,y)}$ 
valid for $x+y>1$ of the reciprocal Beta function to 
\begin{align}
\partial_w R_{\alpha,\beta}(a,b;w)
\big|_{w=0}
&=
\frac{\Gamma(2-\alpha-\beta)}{
\Gamma(2-\alpha)\Gamma(2-\beta)}
F_2\Big(\genfrac{}{}{0pt}{}{2-\alpha-\beta;2,2}{2-\beta,
2-\alpha}\Big|-a,-b\Big)
\nonumber
\\
&- 
\frac{\alpha\beta \Gamma(2-\alpha-\beta)}{
\Gamma(2-\alpha)\Gamma(2-\beta)}
F_2\Big(\genfrac{}{}{0pt}{}{2-\alpha-\beta;1,1}{2-\beta,
2-\alpha}\Big|-a,-b\Big)\;.
\end{align}
Here, $F_2\big(\genfrac{}{}{0pt}{}{a;b_1,b_2}{c_1,c_2}\big|x,y\big)$ 
denotes the second Appell hypergeometric function 
in two variables.
\end{Remark}

\section{Discussion}

\label{sec:discussion}

\subsection{Uniqueness} 

The discussion of possible non-trivial
solutions of the homogeneous Carleman equation in connection with
\eqref{eq:Carl-hom} shows that the solution to \eqref{NLIQ}
is not necessarily unique. However, the expansion in $\lambda$ shows that \eqref{NLIQ} uniquely determines the formal power series expansion of all solutions. It follows that our solution in Theorem~\ref{thm:Lambert} the unique \emph{analytic} solution of \eqref{NLIQ}.

The difference between our solution and any other solution must hence be flat.
Note that flat contributions indeed show up in certain formulae 
for negative coupling constants, see sec.~\ref{sec:warning}.

\subsection{Four dimensions} The solution method should also extend to the
$\lambda\phi^{\star 4}$-model on 4-dimensional Moyal space. For finite
$\CO$ this amounts to changing the integration measure in
\eqref{Gab-integral} from $\td p$ to $p\,\td p$.  This creates much more
severe divergences for $\CO\to \infty$ which require subtle
rescaling by a wavefunction renormalisation $Z(\CO)$ and a more
complicated dependence of $\mu^2$ on $\CO$. Whereas
\eqref{Gab-real} already agrees with \cite[Thm.~4.7]{GW:Phi44nonnon},
up to a global factor $a$ from the changed measure and a global
renormalisation constant, an analogue of \eqref{cottauba} was missing
in \cite{GW:Phi44nonnon}. This lack was compensated by 
a symmetry argument which allowed to prove existence of a solution,
but there was no way to obtain an explicit formula. The methods
developed here give hope to achieve such a formula.

\subsection{Convergence}
We recall that \eqref{HyperInt} and our solution in
Theorem~\ref{thm:Lambert} have non-zero radius of convergence, as
expected for integrable quantum field theory.

\subsection{Integrability} Solving a non-linear problem such as \eqref{NLIQ} by
  (generalised) radicals can only be expected if some deep algebraic
  structure is behind.  We have no idea what it is%
\footnote{%
	  We quote from the bibliographical notes of \cite{Knuth}: \emph{``We find
it a remarkable coincidence that the curves defining the branch cuts of the Lambert W
function (which contain the Quadratix of Hippias) can be used 
not only to square the
circle---which, by proving $\pi$ irrational, Lambert went a long way 
towards proving was impossible by compass and straightedge---but
also to trisect a given angle.''} Is the same capability of Lambert-W used
to solve \eqref{NLIQ}?
},
  but we find it
  worthwhile to explore that connection. We remark that the initial
  action \eqref{action-matrix} is closely related to the action
  $S(\Phi)= V\,\tr (E \Phi^2+\tfrac{\lambda}{3}\Phi^3)$ of the
  Kontsevich model \cite{Kontsevich:1992ti}.  This model gives rise to
  solvable $\lambda\Phi^3$-matricial QFT-models in dimension
  $2,4$ and $6$ \cite{Grosse:2016pob, Grosse:2016qmk} which, however, are
  modest from a number-theoretical point of view: In a perturbative
  expansion of correlation functions only $\log(1+a)$ arises and only
  at lowest order, no polylogarithms as in \eqref{Gab-pert}.  The
  $\Phi^4$-model is much richer and closer to true QFT-models.  The
  Kontsevich model relates to infinite-dimensional Lie algebras and to
  the $\tau$-function of the KdV-hierarchy. It generates intersection
  numbers of stable cohomology classes on the moduli space of complex
  curves \cite{Kontsevich:1992ti}. Something similar should exist for
  $\Phi^4$ as well.

\subsection{Simplicity} Perturbative expansions in realistic quantum field theories like the Standard Model also produce (much more complicated) polylogarithms and other transcendental functions; see for example \cite{ABBFMS:splitting,HennSmirnov:Bhabha,Laporta:g-2at4L,Todorov:GraphicalFunctions,PanzerSchnetz:Phi4Coaction}.
  It would be exciting if the tremendous apparent complexity of those series could also be produced by an integral transform of a simpler function.

\section*{Acknowledgements} 

We are grateful to Spencer Bloch and Dirk Kreimer for invitation to the Les
Houches summer school ``Structures in local quantum field theories''
where the decisive results of this paper were obtained. 
Discussions during this school in particular with 
Johannes Bl\"umlein, David Broadhurst and Gerald Dunne contributed
valuable ideas. RW would like to thank Alexander Hock for pointing out
reference \cite{Gakhov} and Harald Grosse for the long-term
collaboration which preceded this work.

\providecommand{\href}[2]{#2}\providecommand{\eprintlink}[2]{\href{#1}{#2}}

\end{document}